\documentclass[journal]{IEEEtran}
\usepackage{cite}
\usepackage{amsmath,amssymb,amsfonts,dsfont}
\usepackage[noend]{algpseudocode}
\usepackage{algorithmicx}
\usepackage{algorithm}
\usepackage{graphicx}
\usepackage{subfigure}
\usepackage{textcomp}
\usepackage{xcolor}
\usepackage{booktabs}
\usepackage{breqn}
\usepackage{subeqnarray}
\usepackage{cases}
\usepackage{setspace}
\usepackage{amsthm}
\usepackage{bbm}
\usepackage{bm}
\usepackage{lipsum}
\usepackage{enumitem}
\usepackage{mathtools}
\usepackage{url}
\usepackage{makecell}
\usepackage{multirow}

\usepackage{epstopdf}

\allowdisplaybreaks[4]

\newtheorem{thm}{Theorem}

\theoremstyle{remark}

\newtheorem{rem}{Remark}

\def\BibTeX{{\rm B\kern-.05em{\sc i\kern-.025em b}\kern-.08em
    T\kern-.1667em\lower.7ex\hbox{E}\kern-.125emX}}

\begin{document}

\title{Large Language Model-Empowered Decision Transformer for UAV-Enabled Data Collection\\
}
\author{Zhixiong~Chen,~\IEEEmembership{Member,~IEEE},
Jiangzhou~Wang,~\IEEEmembership{Fellow,~IEEE}, \\
Hyundong~Shin,~\IEEEmembership{Fellow,~IEEE},
and Arumugam~Nallanathan,~\IEEEmembership{Fellow,~IEEE}
\thanks{Zhixiong Chen and Arumugam Nallanathan are with the School of Electronic Engineering and Computer Science, Queen Mary University of London, London, U.K. (emails: \{zhixiong.chen, a.nallanathan\}@qmul.ac.uk).}
\thanks{Jiangzhou Wang is with the National Mobile Communications Research Laboratory, Southeast University, Nanjing, China (email: j.z.wang@seu.edu.cn).}
\thanks{Hyundong Shin is with the Department of Electronics and Information Convergence Engineering, Kyung Hee University, Yongin-si, Gyeonggido 17104, Republic of Korea (e-mail: hshin@khu.ac.kr).}
}

\maketitle

\begin{abstract}
The deployment of unmanned aerial vehicles (UAVs) for reliable and energy-efficient data collection from spatially distributed devices holds great promise in supporting diverse Internet of Things (IoT) applications. Nevertheless, the limited endurance and communication range of UAVs necessitate intelligent trajectory planning. While reinforcement learning (RL) has been extensively explored for UAV trajectory optimization, its interactive nature entails high costs and risks in real-world environments. Offline RL mitigates these issues but remains susceptible to unstable training and heavily rely on expert-quality datasets. To address these challenges, we formulate a joint UAV trajectory planning and resource allocation problem to maximize energy efficiency of data collection. The resource allocation subproblem is first transformed into an equivalent linear programming formulation and solved optimally with polynomial-time complexity. Then, we propose a large language model (LLM)-empowered critic-regularized decision transformer (DT) framework, termed LLM-CRDT, to learn effective UAV control policies. In LLM-CRDT, we incorporate critic networks to regularize the DT model training, thereby integrating the sequence modeling capabilities of DT with critic-based value guidance to enable learning effective policies from suboptimal datasets. Furthermore, to mitigate the data-hungry nature of transformer models, we employ a pre-trained LLM as the transformer backbone of the DT model and adopt a parameter-efficient fine-tuning strategy, i.e., LoRA, enabling rapid adaptation to UAV control tasks with small-scale dataset and low computational overhead. Extensive simulations demonstrate that LLM-CRDT outperforms benchmark online and offline RL methods, achieving up to 36.7\% higher energy efficiency than the current state-of-the-art DT approaches.
\end{abstract}

\section{Introduction}
The rapid development of the Internet of Things (IoT) has driven the emergence of innovative applications such as smart cities and environmental monitoring, reshaping diverse aspects of daily life and industrial operations.These applications rely on large-scale deployments of spatially distributed IoT devices, such as sensors, to continuously monitor areas of interest and collect valuable data for analysis and decision-making \cite{7496795}. Reliable data collection from these dispersed sensing nodes is therefore critical to unlocking the full potential of IoT systems, particularly in regions where fixed network infrastructure is impractical or prohibitively costly. In this context, unmanned aerial vehicles (UAVs) offer a compelling solution by leveraging their deployment flexibility and precise manoeuvrability to complement ground networks and deliver ubiquitous connectivity for IoT devices \cite{8660516}. Serving as aerial mobile base stations, UAVs can efficiently harvest data from distributed devices and relay it to cloud or edge servers for further processing, making them well-suited for mission-critical applications such as disaster response and precision farming.

Despite the aforementioned advantages, the limited flight endurance and communication range of UAVs pose critical challenges for IoT data collection. Towards this end, extensive research has focused on optimizing UAV trajectory planning and wireless resource management to balance coverage, latency, and energy efficiency throughout the mission. Specifically, in \cite{9444660}, an energy-constrained UAV was dispatched to collect data from a sparse IoT sensor network, where approximation and heuristic algorithms were designed to optimize the UAV's hovering positions and durations for maximizing the collected data volume. In \cite{9957134} and \cite{8432487}, an iterative algorithm and a dynamic programming-based approach were respectively proposed to minimize data collection completion time through jointly optimizing UAV's trajectory and wireless resource allocation. By employing convex optimization techniques for UAV trajectory design, the work in \cite{8779596} minimized the data collection time in a multi-UAV setting. In \cite{10007855}, an ant colony algorithm was developed to minimize the age of information of data collected from ground sensor nodes. Considering the security threats in UAV-aided wireless networks, a short-packet transmission scheme was proposed in \cite{10044099} to ensure both the freshness and security of collected data, wherein a successive convex optimization-based approach was employed to optimize the UAV's trajectory. In \cite{9826413} and \cite{9562293}, convex stochastic programming and fractional programming approaches were respectively investigated to optimize the UAV's trajectory and maximize the energy efficiency of data collection tasks. Although these traditional optimization and heuristic-based offline methods effectively improved data collection performance, they typically require perfect channel state information over the entire time horizon. This hinders their practical deployment, as IoT networks usually operate in highly dynamic wireless environments where communication conditions change unpredictably.

To address the limitations of traditional methods, reinforcement learning (RL) has been widely employed to learn UAV trajectory planning and resource allocation policies through interactions with the environment, eliminating the need for prior knowledge of channel state information and other environmental dynamics \cite{sutton1998reinforcement, 10706837, 9862981}. Specifically, deep Q-network-based approaches were developed in \cite{9701330, 9385412} to learn UAV flight path planning policies, where the flight path was discretized and the UAV only required to select actions from a finite action space. To enable more fine-grained control, the deep deterministic policy gradient (DDPG) algorithm was adopted in \cite{9455139, 9801656} to directly optimize the UAV control policies within a continuous action space, thereby enhancing manoeuvrability. Moreover, the twin delayed deep deterministic policy gradient (TD3) \cite{pmlr-v80-fujimoto18a} builds upon DDPG by incorporating twin critics, delayed policy updates, and target policy smoothing, leading to more accurate value estimation and reduced variance. Accordingly, TD3-based UAV control policies was proposed in \cite{9426899}, achieving a more stable learning process and higher data collection rewards than DDPG. To address the low exploration efficiency caused by the deterministic policies in DDPG and TD3, soft actor-critic (SAC)-based UAV control approaches was proposed in \cite{10086052, 10564117}, employing a stochastic policy to facilitate better exploration in wireless environments with sparse rewards and complex dynamics. Despite their flexibility and effectiveness, the above approaches are all online RL methods that operate in a trial-and-error manner and require extensive real-time interactions with the environment. This makes them less practical for UAV-enabled IoT data collection due to the high operational costs and risks associated with continuous exploration. Moreover, the sample inefficiency of online RL demands a large number of interactions before converging to a high-performing policy, further limiting their applicability in large-scale IoT scenarios.

To tackle the limitations of online RL, offline RL offers a promising alternative by learning effective policies entirely from pre-collected datasets. This paradigm enables the use of historical flight logs, network operation records, and simulated data to train UAV control and wireless network optimization policies, thereby avoiding the risks and costs inherent to real-time trial-and-error \cite{NEURIPS2021_a8166da0, NEURIPS2021_3d3d286a}. However, offline RL algorithms are susceptible to value overestimation, particularly for out-of-distribution actions, and often experience unstable training due to their reliance on bootstrapping to propagate returns. Recently, decision transformer (DT) \cite{NEURIPS2021_7f489f64} was introduced to cast the offline RL problem as a sequence modeling task, effectively converting it into a supervised learning problem. By exploiting the powerful sequence representation and long-range dependency modeling capabilities of transformer architectures, DT enables the direct prediction of actions conditioned on desired returns, thereby improving learning stability and mitigating value overestimation. Inspired by the advantages of DT, the work in \cite{10848143} made an initial attempt by proposing a DT-based approach to jointly optimize UAV positioning and computation resource allocation in a UAV-assisted edge computing network, effectively improving fairness-based throughput.

Despite the impressive performance of DT, it still faces several challenges when applied to UAV-enabled IoT data collection scenarios. Firstly, DT lacks the stitching ability, and its performance is inherently constrained by the quality of the underlying dataset. When the training data is suboptimal, DT struggles to combine actions from suboptimal trajectories into optimal policies, thereby necessitating expert-level datasets to achieve satisfactory performance. Moreover, The transformer architecture is well known to be data-hungry, requiring massive amounts of training data to achieve satisfactory performance. However, collecting such large-scale and high-quality training datasets is often difficult and costly in UAV-enabled wireless networks. To this end, we propose a critic-regularized DT training paradigm that integrates the sequence modeling strengths of DT with the policy improvement guidance of critic networks, enabling the learning of effective UAV control policies from suboptimal datasets and removing the dependence on expert-level training dataset. In addition, inspired by the impressive few-shot generalization ability of pre-trained large language models (LLMs), we employ a pre-trained LLM as the transformer backbone of the DT and adopt a parameter-efficient fine-tuning approach, namely LoRA \cite{hu2022lora}, to rapidly adapt the DT model for UAV control. This design not only reduces data and computational requirements but also facilitates the acquisition of high-quality policies without relying on massive training data.
Our main contributions are summarized as follows:
\begin{itemize}
  \item We formulate a UAV trajectory planning and resource allocation problem to maximize the energy efficiency of UAV-assisted data collection. Firstly, the resource allocation subproblem is converted into an equivalent linear programming, which admits an optimal solution in polynomial time. The UAV trajectory planning task is then reformulated as an offline RL problem, and we provide a theoretical analysis of the limitations of existing DT-based methods in solving it.

  \item To solve the UAV trajectory planning problem and address the limitations of DT methods, we propose an LLM-empowered critic-regularized DT approach, termed LLM-CRDT. Specifically, we employ two critic networks to estimate the state-action values. Unlike standard DT approaches that rely solely on supervised learning over return-conditioned trajectories, our method introduces a state-action value-based regularization term into the training loss. This encourages the DT model to select actions that align with high-value predictions from the critics, thereby improving policy performance beyond that of the behavior trajectories in the offline dataset.

  \item Instead of training LLM-CRDT from scratch, we employ a pre-trained LLM as the transformer backbone of the DT model and fine-tune LLM-CRDT with LoRA, where the pre-trained LLM remains frozen while trainable low-rank matrices are injected into each transformer block. Fine-tuning updates only the input data encoder, LoRA adapters, action decoder, and critic networks, accounting for only a small fraction of the total model parameters. This not only reduces data and computational requirements but also leverages the universal reasoning capabilities of pre-trained LLMs to achieve effective UAV control.

  \item We conduct extensive simulations to evaluate the effectiveness of the proposed LLM-CRDT and resource allocation approach. The results show that our method outperforms benchmark online and offline RL approaches. Compared to the current state-of-the-art DT method, it achieves up to a 36.7\% improvement in energy efficiency for UAV data collection missions. Moreover, the results demonstrate that our approach is capable of learning effective policies from suboptimal datasets.
\end{itemize}

The remainder of this paper is organized as: Section \ref{sec:system_model} introduces the system model. In Section \ref{sec:optimal_RA}, we present the optimal resource allocation algorithm, then reformulates the UAV control problem as an offline RL problem. Section \ref{sec:LLM_CRDT} illustrates the proposed LLM-CRDT to learn effective UAV control policies. Simulation results are presented in Section \ref{sec:simulations}. Finally, Section \ref{sec:conclusion} concludes this work.

\section{System Model}\label{sec:system_model}
As shown in Fig. \ref{fig:sys_model}, this work considers a UAV-enabled IoT network, where a rotary-wing UAV is dispatched to collect data from $N$ distributed IoT devices within a rectangular region of side lengths $X_{\max}$ and $Y_{\max}$, thereby supporting intelligent applications such as environmental monitoring and preventive maintenance. Following prior works, e.g., \cite{9862981, 9701330, 9385412, 9455139, 9801656}, the system runs in discrete time horizon consisting of $T$ slots, each with duration $\delta$. The IoT devices are indexed by $\mathcal{N}=\{ 1, 2, \cdots, N \}$. Each device $n$ ($n\in \mathcal{N}$) is associated with certain amount of data $D_n$ that required to be collected in real-time, and its location information is denoted using three dimensional (3D) Cartesian coordinates, i.e., $\bm{l}_n = (x_n, y_n, 0)$, where $x_n$ and $y_n$ are its coordinates along the $x$- and $y$-axes, respectively.

\subsection{UAV Movement Model}
In this work, the UAV is assumed to fly at a constant altitude $H$ ($H>0$) throughout the entire data collection mission \cite{8779596, 10007855}. It is worth noting that the following proposed approaches can be readily extended to variable-altitude trajectories by introducing an additional control variable for vertical movement.
The UAV's position in time slot $t$ is represented by $\bm{l}_t^{\rm{u}} = (x_t^{\rm{u}}, y_t^{\rm{u}}, H)$, where $0 \le x_t^{\rm{u}} \le X_{\max}$ and $0 \le y_t^{\rm{u}} \le Y_{\max}$.
Each time slot $t$ is divided into two parts, i.e., $\delta = \delta_t^{\rm{fly}} + \delta_t^{\rm{hover}}$, where $\delta_t^{\rm{fly}}$ is the UAV's flight time and $\delta_t^{\rm{hover}}$ is the hovering time for data collection.
At the start of each slot $t$, the UAV determines the flight duration $\delta_t^{\rm{fly}}$ ($0 \le \delta_t^{\rm{fly}} \le \delta$) and flight direction $\theta_t$ ($0 \le \theta_t \le 2\pi$). It then travels to a new position at an average speed $v_t$ ($0 \le v_t \le V_{\max}$), where $V_{\max}$ denotes the maximum allowable flight speed determined by the UAV's hardware limitations.
Thus, the UAV's horizontal location in slot $t$ is calculated as follows:
\begin{align}
\left\{ {\begin{array}{*{20}{c}}
x_t^{\rm{u}} = x_{t-1}^{\rm{u}} + v_t \delta_t^{\rm{fly}} \cos(\theta_t),\\
y_t^{\rm{u}} = y_{t-1}^{\rm{u}} + v_t \delta_t^{\rm{fly}} \sin(\theta_t).
\end{array}} \right.
\end{align}

After reaching the desired position, the UAV hovers there for the remaining time in slot $t$, i.e., $\delta_t^{\rm{hover}}$, to collect data from the IoT devices.
According to the energy consumption model of UAV \cite{10706837, 9701330}, the propulsion power consumption of the UAV with flight speed $v$ is
\begin{align}
P(v) &= P_1\Big(1 + \frac{3 v^2}{U_{\rm{tip}}^2} \Big) + P_2 \Big(\sqrt{1 + \frac{v^4}{4 v_0^4}}  - \frac{v^2}{2 v_0^2} \Big)^{\frac{1}{2}} \notag \\[-0.1cm]
&+ \frac{1}{2} d_0\rho g A v^3,
\end{align}
where the three terms on the right-hand side correspond to the blade profile power, induced power, and parasite power, respectively.
$P_1$ denotes the blade profile power during hovering, and $U_{\text{tip}}$ is the tip speed of the rotor blade. $P_2$ and $v_0$ represent the induced power and the mean rotor induced velocity in hovering, respectively. For the parasite power, $d_0$, $\rho$, $g$, and $A$ denote the fuselage drag ratio, air density, rotor solidity, and rotor disc area, respectively.
By setting $v=0$, the hovering power consumption of the UAV is given by $P_1 + P_2$.
Therefore, the energy consumption of the UAV in time slot $t$ can be expressed as
\begin{align}
E_t = \delta_t^{\rm{fly}}P(v_t) + \delta_t^{\rm{hover}}(P_1 + P_2).
\end{align}

\begin{figure}[t]
\centering
\includegraphics[width=0.37\textwidth]{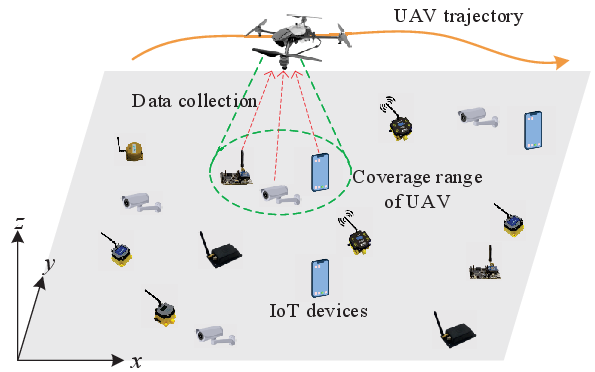}\\
\caption{The considered UAV-enabled IoT networks.}
\label{fig:sys_model}
\end{figure}

\subsection{Channel Model}
This work adopts the widely used probabilistic line-of-sight (LoS) channel model \cite{10007855, 10706837} to characterize the large-scale attenuation of the wireless transmission links between the UAV and IoT devices. The probability of establishing a geometrical LoS connection depends on both environment-specific parameters and the UAV's elevation angle relative to each device.
Specifically, the LoS probability for device $n$ in slot $t$ is given by
\begin{align}
P_{n,t}^{\rm{LoS}} = \frac{1}{1 + a\exp \left(- b(\omega_{n,t} - a) \right)},
\end{align}
where $a$ and $b$ are environment-related parameters, $\omega_{n,t}$ denotes the elevation angle between the UAV and device $n$ in time slot $t$, i.e.,
\begin{align}
\omega_{n,t} = \frac{180}{\pi} \arctan \Big( \frac{H}{\sqrt{(x_t^{\rm{u}} - x_n)^2 + (y_t^{\rm{u}} - y_n)^2} } \Big).
\end{align}
Accordingly, the non-line-of-sight (NLoS) channel probability of device $n$ in slot $t$ is $P_{n,t}^{\rm{NLoS}} = 1 - P_{n,t}^{\rm{LoS}}$.
The expected channel gain of device $n$ in slot $t$ is
\begin{align}
g_{n,t} = \frac{g_0\left(P_{n,t}^{\rm{LoS}} + (1 - P_{n,t}^{\rm{LoS}}) \kappa  \right)}{\left( {H^2 + ( x_t^{\rm{u}} - x_n)^2 + (y_t^{\rm{u}} - y_n)^2} \right)^{\frac{\iota}{2}}},
\end{align}
where $\iota$ is the path loss exponent, $g_0$ is the channel gain at the reference distance $d_0=1$ m, $0<\kappa < 1$ captures the additional attenuation incurred under NLoS conditions.

\subsection{Data Collection Model}
We consider an OFDMA-based network \cite{5281762, 6094142} in which $M$ resource blocks (RBs) or chunks support device-to-UAV data transmission, the set of RB indices is $\mathcal{M}=\{1,2,\cdots,M\}$. In line with prior studies, such as \cite{10367814, 10444714}, we eliminate inter-device interference by restricting the allocation so that each RB is assigned to at most one device, while each device uses only one RB per slot to communicate with the UAV.
Let $z_{n,m}^{(t)} \in \{0, 1\}$ denote the RB allocation decision for device $n$ in time slot $t$, where $z_{n,m}^{(t)} =1$ indicates that RB $m$ is allocated to device $n$, and $z_{n,m}^{(t)} =0$ otherwise. We represent $\bm{Z}_t = \{ z_{n,m}^{(t)}: n \in \mathcal{N}, m \in \mathcal{M}\}$ as the RB assignments in slot $t$. Denoting by $p_n$ the transmit power of device $n$, its achievable transmit rate on RB $m$ in slot $t$ is
\begin{align}
r_{n,m}^{(t)} = W \log_2\Big( 1 + \frac{p_n g_{n,t}}{I_{m,t} + W N_0} \Big),
\end{align}
where $W$ is the per-RB bandwidth, $N_0$ is the noise power spectral density, $I_{m,t}$ represents the co-channel interference from other services sharing the same frequency resources.

Due to the limited transmit power of IoT devices and substantial path loss, the UAV can only collect data from devices located within its coverage area during the hovering duration of each time slot $t$.
Let $\mathcal{N}_t^{\rm{c}}$ represent the set of devices covered by the UAV in time slot $t$, which is defined as
\begin{align}
\mathcal{N}_t^c = \left\{ n \left| d_t^n \le R_{\max}, n \in \mathcal{N} \right. \right\},
\end{align}
where $d_t^n = \sqrt{( x_n - x_t^u)^2 + (y_n - y_t^u)^2}$ is the horizontal distance between device $n$ and the UAV in slot $t$.
$R_{\max}$ is the maximal horizontal coverage range of the UAV decided by the maximal azimuth angle $\omega_{\max}$. Thus, $R_{\max}$ is given by
\begin{align}
R_{\max} = H\tan \left( \omega_{\max} \right).
\end{align}

Restricted by limited wireless resources, the UAV may not collect data from all devices within its coverage area during each time slot $t$. Thus, it is crucial to judiciously select the devices for data collection in each slot.
Let $\beta_{n,t} \in \{0,1\}$ denote the selection flag for device $n$ in slot $t$, where $\beta_{n,t}=1$ if device $n$ is chosen by the UAV for data collection, and $\beta_{n,t}=0$ otherwise.
Clearly, $\beta_{n,t}=0$ if $n \notin \mathcal{N}_t^{\rm{c}}$. Let $\bm{\beta}_t=\{\beta_{n,t}: n \in \mathcal{N}\}$ represent the device selection decision in slot $t$.
For each device $n$, the amount of data collected during slot $t$ is expressed as:
\begin{align}
\Delta_{n, t}^{\rm{collect}} = \min \Big\{ \beta_{n,t} \sum \nolimits_{m=1}^M z_{n,m}^{(t)} r_{n,m}^{(t)} \delta_t^{\rm{hover}}, ~D_{n,t} \Big\},
\end{align}
where $D_{n,t}$ denotes the remaining data on device $n$ at the beginning of slot $t$, it evolves as:
\begin{align}
D_{n,t+1} = \max \left\{ D_{n,t} - \Delta_{n,t}^{\rm{collect}}, ~0 \right\},
\end{align}
where $D_{n,0}=D_n$.
Accordingly, the energy efficiency of the UAV in slot $t$ can be defined as
\begin{align}
\phi_t = \frac{\sum\nolimits_{n = 1}^N \Delta_{n,t}^{\rm{collect}}}{E_t},
\end{align}
which quantifies data successfully gathered per unit energy in slot $t$. A higher value of $\phi_t$ indicates more energy-efficient data collection.

\subsection{Problem Formulation}
This work focuses on maximizing the UAV's energy efficiency during the data collection mission by co-designing the UAV control and resource allocation policies. The optimization problem is formulated as follows:
\begin{align}
&\max_{\left\{ \theta_t, v_t, \delta_t^{\rm{fly}}, \bm{\beta}_t, \bm{Z}_t \right\}_{t = 1}^{T}}  ~ \sum\limits_{t = 1}^T \frac{\sum\nolimits_{n = 1}^N \Delta_{n,t}^{\rm{collect}}}{E_t} \label{prob:P}\\
\text{s.~t.~~} & 0 \le \theta_t \le 2\pi, \forall 1 \le t \le T, \label{cons:P_1}\tag{\theequation a}\\
& 0 \le v_t \le V_{\max}, \forall 1 \le t \le T, \label{cons:P_2}\tag{\theequation b} \\
& 0 \le  \delta_t^{\rm{fly}} \le \delta, \forall 1 \le t \le T, \label{cons:P_3}\tag{\theequation c} \\
& 0 \le x_t^{\rm{u}} \le X_{\max}, \forall 1 \le t \le T, \label{cons:P_4}\tag{\theequation d} \\
& 0 \le y_t^{\rm{u}} \le Y_{\max}, \forall 1 \le t \le T, \label{cons:P_5}\tag{\theequation e} \\
& 0 \le \beta_{n,t} \le \mathbbm{1}(n \in \mathcal{N}_t^{\rm{c}}),  \forall n \in \mathcal{N}, 1 \le t \le T, \label{cons:P_6}\tag{\theequation f} \\
& \beta_{n,t} \in \{0, 1\}, \forall n \in \mathcal{N}, 1 \le t \le T, \label{cons:P_7}\tag{\theequation g} \\
& \sum\nolimits_{n = 1}^N z_{n,m}^{(t)} \le 1, \forall m \in \mathcal{M}, 1 \le t \le T, \label{cons:P_8}\tag{\theequation h} \\
& \sum\nolimits_{m = 1}^M z_{n,m}^{(t)} \le 1, \forall n \in \mathcal{N}, 1 \le t \le T, \label{cons:P_9}\tag{\theequation i} \\
& z_{n,m}^{(t)} \in \{0, 1\}, \forall n \in \mathcal{N}, m \in \mathcal{M}, 1 \le t \le T, \label{cons:P_10}\tag{\theequation j}
\end{align}
wherer constraints \eqref{cons:P_1}-\eqref{cons:P_3} correspond to the restrictions of UAV control policy, ensuring that the UAV operates within its allowable limits.
\eqref{cons:P_4} and \eqref{cons:P_5} enforce the UAV's position to remain within the predefined operational region.
\eqref{cons:P_6} and \eqref{cons:P_7} ensure that only devices within the UAV's coverage area can be selected for data collection in each slot. Here, $\mathbbm{1}(\cdot)$ denotes the indicator function, which returns 1 if and only if the condition inside the parentheses is satisfied, and 0 otherwise.
\eqref{cons:P_8}, \eqref{cons:P_9} and \eqref{cons:P_10} are the resource allocation restrictions.

Directly solving problem \eqref{prob:P} requires complete channel state information between all devices and the UAV across all time slots and for all possible UAV positions. However, this is impractical due to the continuous nature of the UAV's position space, which contains infinite possible locations. Moreover, accurate channel state information for future time slots cannot be obtained in advance.
To address these challenges, we propose efficient approaches to solve problem \eqref{prob:P} in the following sections, without relying on future channel state information of the device-UAV links.

\section{Optimal Resource Allocation and Problem Transformation}\label{sec:optimal_RA}
This section proposes an effective approach to derive the optimal device selection and resource allocation policy, and reformulate problem \eqref{prob:P} as a Markov decision process (MDP) for UAV control. We then highlight the limitations of existing RL and DT algorithms, which motivate the development of the proposed LLM-CRDT to effectively learn UAV control policies in offline settings, as detailed in Section \ref{sec:LLM_CRDT}.

\subsection{Optimal Device Selection and Resource Allocation}\label{subsec:resource_allocation}
According to the formulation of problem \eqref{prob:P}, for any given UAV control policy in arbitrary time slot $t$ ($1 \leq t \leq T$), the device selection and resource allocation decisions in slot $t$ affect only the energy efficiency of this slot and are independent of other slots. Thus, the device selection and resource allocation problem can be decoupled across time slots and solved independently. Accordingly, we focus on deriving the optimal device selection and resource allocation strategy for a specific time slot $t$ under an arbitrary UAV control policy. Notably, the proposed algorithm can be directly applied to any time slot.
Since $\beta_{n,t} = \sum_{m=1}^M z_{n,m}^{(t)}$, once RBs are assigned the device selection is implied.
Accordingly, we decouple the slot-$t$ resource allocation subproblem from \eqref{prob:P} as:
\begin{align}
\max_{\bm{Z}_t} & ~ \sum\nolimits_{n = 1}^N \min \Big\{ \sum \nolimits_{m=1}^M z_{n,m}^{(t)} r_{n,m}^{(t)}  \delta_t^{\rm{hover}}, ~D_{n,t} \Big\} \label{prob:P_device}\\[-0.1cm]
\text{s.~t.~~} & \eqref{cons:P_8}, \eqref{cons:P_9}, \eqref{cons:P_10}. \notag
\end{align}

Directly solving \eqref{prob:P_device} is challenging due to its nonlinear integer structure. We instead cast it as a maximum-weight perfect matching, which can be solved in polynomial time. We therefore convert it to a maximum-weight perfect matching problem, which admits a polynomial-time solution.
Specifically, we define a bipartite graph $\bm{G}=(\mathcal{V}, \mathcal{E})$, where the vertex set is defined as $\mathcal{V}=\mathcal{N} \cup \bar{\mathcal{M}}$, and $\mathcal{E}$ represents the set of edges connecting vertices in $\mathcal{N}$ and $\bar{\mathcal{M}}$. In the graph $\bm{G}$, each vertex $n \in \mathcal{N}$ corresponds to device $n$, while $\bar{\mathcal{M}} = \mathcal{M} \cup \mathcal{M}_v$ is an extended version of $\mathcal{M}$, where each vertex $m \in \mathcal{M}$ corresponds to RB $m$. The set $\mathcal{M}_v$ contains virtual vertices introduced to ensure $\left|\bar{\mathcal{M}} \right| = \left|\mathcal{N} \right|$, thereby forming a balanced bipartite graph. The weight of each edge $(n, m) \in \mathcal{E}$ is
\begin{align}\label{eq:edge_weight}
\!\!\!\Omega_{n,m}^{(t)} \!=\! \left\{ {\begin{array}{*{20}{c}}
\!\!\!\min \!\big\{r_{n,m}^{(t)}  \tau_t^{\rm{hover}}, D_{n,t} \big\}, \!\!&\text{if}~ n \in\! \mathcal{N}_t^{\rm{c}}, m \in\! \mathcal{M},\\
0,& \text{otherwise}.
\end{array}} \right.
\end{align}
Based on graph $\bm{G}$, we reformulate \eqref{prob:P_device} as a maximum weight perfect bipartite matching problem, where the objective is to seek a perfect matching $\bm{R}$ maximizing $\sum_{(n, m) \in \bm{R}} \Omega_{n,m}^{(t)}$.
We introduce the binary connector $c_{n,m} \in \{0, 1\}$ for edge $(n, m)$: $c_{n,m} =1$ represents assigning RB $m$ to device $n$, $c_{n,m} =0$ otherwise. Hence, the bipartite matching problem is formulated as follows:
\begin{align}
\max_{\{c_{n,m}: n \in \mathcal{N}, m \in \bar{\mathcal{M}} \}} & ~ \sum\nolimits_{n = 1}^N \sum\nolimits_{m = 1}^{\left| \bar{\mathcal{M}} \right|} c_{n,m} \Omega_{n,m}^{(t)} \label{prob:P_device_matching}\\
\text{s.~t.~~} & \sum\nolimits_{n = 1}^N c_{n,m} = 1, \forall m \in \bar{\mathcal{M}}, \label{cons:P_dm1}\tag{\theequation a}\\
& \sum\nolimits_{m = 1}^{\left| \bar{\mathcal{M}} \right|} c_{n,m} = 1, \forall n \in \mathcal{N}, \label{cons:P_dm2}\tag{\theequation b}\\
& c_{n,m} \in \left\{0,1 \right\}, \forall m \in \bar{\mathcal{M}}, n \in \mathcal{N}. \label{cons:P_dm3}\tag{\theequation c}
\end{align}
However, finding the optimal solution to problem \eqref{prob:P_device_matching} remains non-trivial, as it is an integer programming problem. By relaxing the integer constraint \eqref{cons:P_dm3}, problem \eqref{prob:P_device_matching} can be transformed into the following linear programming:
\begin{align}
\max_{\{c_{n,m}: n \in \mathcal{N}, m \in \bar{\mathcal{M}}\}_{n=1}^{N}} & ~ \sum\nolimits_{n = 1}^N \sum\nolimits_{m = 1}^{\left| \bar{\mathcal{M}} \right|} c_{n,m} \Omega_{n,m}^{(t)} \label{prob:P_device_matching1}\\
\text{s.~t.~~} & \eqref{cons:P_dm1}, \eqref{cons:P_dm2}, \notag\\
& 0 \le c_{n,m} \le 1, \forall m \in \bar{\mathcal{M}}, n \in \mathcal{N}. \label{cons:P_dm3_1}\tag{\theequation a}
\end{align}
\begin{rem}\label{rem:one}
In \eqref{prob:P_device_matching1}, each row associated with \eqref{cons:P_dm1} and \eqref{cons:P_dm2} has exactly one coefficient equal to \(1\), implying total unimodularity of the coefficient matrix (since every square submatrix has determinant \(0\), \(1\), or \(-1\)).
According to \cite{schrijver2003combinatorial}, \eqref{prob:P_device_matching1} attains an integral optimum that is also optimal for the original problem \eqref{prob:P_device_matching}.
\end{rem}

Remark \ref{rem:one} guarantees that solving \eqref{prob:P_device_matching1} directly recovers the optimum of \eqref{prob:P_device_matching}. This work employs the current matrix multiplication time algorithm \cite{10.1145/3424305} to find the optimal solution of problem \eqref{prob:P_device_matching1}, which has a time complexity of $\mathcal{O}((N^{2+1/6})^2)$ as problem \eqref{prob:P_device_matching1} involves $N^2$ variables.

Based on the above analysis, problem \eqref{prob:P_device} can be efficiently solved to obtain the optimal device selection and resource allocation decisions. For clarity, we present the detailed solution procedure in Algorithm \ref{alg:RB_allocation}.
Firstly, we transform problem \eqref{prob:P_device} into its equivalent bipartite graph matching form, i.e., \eqref{prob:P_device_matching}, which involves calculating the edge weights and verifying the feasibility of assigning each RB to each device. This incurs a time complexity of $\mathcal{O}(NM)$.
Hence, we convert \eqref{prob:P_device_matching} into its equivalent problem, i.e., \eqref{prob:P_device_matching1}, and find the optimal resource allocation decision $\bm{Z}_t^*$. The optimal device selection policy is then computed as $\beta_{n,t}^*=\sum_{m=1}^M z_{n,m}^{(t), *}$.
Consequently, the overall time complexity of Algorithm \ref{alg:RB_allocation} is $\mathcal{O}(NM + (N^{2+1/6})^2)$.

\begin{algorithm}[!t] \small
\caption{Optimal Device Selection and RB Allocation}
\label{alg:RB_allocation}
\begin{spacing}{0.85}
\begin{algorithmic}[1]
\State Formulate the RB allocation problem \eqref{prob:P_device}
\State Construct a bipartite graph $\bm{G}=(\mathcal{V}, \mathcal{E})$ and compute the weights of each edge in $\mathcal{E}$ based on \eqref{eq:edge_weight}.
\State Construct the linear programming problem of the bipartite graph matching problem, i.e., \eqref{prob:P_device_matching1}
\State Solve problem \eqref{prob:P_device_matching1} and obtain the optimal bipartite perfect matching $\{\bm{c}_n\}_{n=1}^N$
\State Calculate the optimal RB allocation decision as $\bm{Z}_t^*=\{z_{n,m}^{(t), *}=c_{n,m}: n \in \mathcal{N}, m \in \mathcal{M}\}$.
\State Calculate the optimal device selection decision as $\bm{\alpha}_t^*=\{\alpha_{n,t}^*=\sum_{m=1}^M z_{n, m}^{(t), *}: \forall n \in \mathcal{N}\}$
\State Return the optimal device selection decision $\bm{\beta}_t^*$ and the optimal RB allocation decision $\bm{Z}_t^*$
\end{algorithmic}
\end{spacing}
\end{algorithm}

\subsection{Problem Transformation and Analysis}\label{subsec:prob_trans}
Based on the above analysis, the optimal device selection $\bm{\beta}_t^*=\{\beta_{n,t}^*: n \in \mathcal{N}\}$ and optimal resource allocation policy $\bm{Z}_t^*=\{z_{n,m}^{(t), *}: n \in \mathcal{N}, m \in \mathcal{M}\}$ under any given UAV control policy can now be obtained using Algorithm \ref{alg:RB_allocation}. By substituting $\bm{\beta}_t^*$ and $\bm{Z}_t^*$ into the original problem \eqref{prob:P}, we reduce it to the following UAV control problem:
\begin{align}
\!\!\!\!\max_{\left\{ \theta_t, v_t, \delta_t^{\rm{fly}}\right\}_{t = 1}^{T}} & \!\!\sum\limits_{t = 1}^T \!\frac{\sum\limits_{n = 1}^N \!\min \!\Big\{\!\sum \limits_{m=1}^M z_{n,m}^{(t),*} r_{n,m}^{(t)}  \delta_t^{\rm{hover}}, D_{n,t} \Big\} }{E_t} \label{prob:P_UAV}\\
\text{s.~t.~~} & \eqref{cons:P_1}, \eqref{cons:P_2}, \eqref{cons:P_3}, \eqref{cons:P_4}, \eqref{cons:P_5}. \notag
\end{align}

Problem \eqref{prob:P_UAV} is intrinsically a sequential decision-making problem, which is impractical to solve directly as it requires precise channel state information for all device-UAV links over the entire time horizon and across all possible UAV positions. Moreover, even with full channel state information, the long-term and non-convex nature of \eqref{prob:P_UAV} would still lead to significant computational challenges.
To tackle these issues, we reformulate problem \eqref{prob:P_UAV} as a MDP and employ RL algorithms to learn a satisfactory UAV control policy.
Specifically, we define the MDP as a tuple $\left(\mathcal{S}, \mathcal{A}, \mathcal{P}, \mathcal{R}, \gamma \right)$: $\mathcal{S}$ is the state space, $\mathcal{A}$ is the action space, $\mathcal{P}: \mathcal{S}\times \mathcal{A} \rightarrow \mathcal{S}$ is the transition model, $\mathcal{R}: \mathcal{S}\times \mathcal{A} \rightarrow \mathbb{R}$ is the reward function, $\gamma \in [0, 1)$ is the discount factor. In each time slot $t$, one can choose an action $\bm{a}_t \in \mathcal{A}$ based on the current state $\bm{s}_t \in \mathcal{S}$. This action is then applied to the environment, resulting in a transition to $\bm{s}_{t+1} \in \mathcal{S}$ and returning $r_t = \mathcal{R}(s_t, a_t)$.
Specifically, the state, action, and reward function are designed as follows:
\begin{itemize}
  \item \emph{State}: In each time slot $t$, the state $\bm{s}_t$ consists of the horizontal position of the UAV, the remaining data amount of devices, and the coverage indicator of devices, i.e.,
  \begin{align}
  \bm{s}_t \!=\! (x_t^{\rm{u}}, y_t^{\rm{u}}, \!D_{1,t}, \cdots, D_{N,t}, \!\mathbbm{1}(1 \in \mathcal{N}_t^{\rm{c}}), \!\cdots\!, \!\mathbbm{1}(N \in \mathcal{N}_t^{\rm{c}})).
  \end{align}

  \item \emph{Action}: The UAV control action in each time slot $t$, i.e., $\bm{a}_t \in \mathcal{A}$, comprises the flight direction, speed, and flight time, defined as
  \begin{align}
  \bm{a}_t = (\theta_t, v_t, \delta_t^{\rm{fly}}).
  \end{align}
  Note that $\bm{a}_t$ must satisfy the constraints \eqref{cons:P_1}-\eqref{cons:P_5} in problem \eqref{prob:P_UAV}.

  \item \emph{Reward}: According to the objective of problem \eqref{prob:P_UAV}, we set the per-slot reward to the energy efficiency, i.e.,
  \begin{align}
   \!r_t = \frac{\sum\limits_{n = 1}^N \min \Big\{\sum \limits_{m=1}^M z_{n,m}^{(t),*} r_{n,m}^{(t)} \delta_t^{\rm{hover}}, D_{n,t} \Big\} }{E_t}.
  \end{align}
\end{itemize}

Based on the above defined MDP, various online RL approaches, such as SAC \cite{pmlr-v80-haarnoja18b} and TD3 \cite{pmlr-v80-fujimoto18a}, can be utilized to learn effective UAV control policies.
These algorithms continuously interact with the environment and update their policy to maximize the long-term cumulative reward, i.e., $\sum_{t=1}^T \gamma^t r_t$.
However, their interactive nature limits practical deployment, particularly for UAV control tasks where real-time exploration is often infeasible or costly. In addition, online RL requires persistent access to the environment for hyperparameter tuning and policy refinement, which is both resource- and time-intensive.
Although offline RL methods, such as temporal-difference (TD) learning \cite{sutton1998reinforcement}, can learn effective policies entirely from pre-collected dataset without interacting with the environment, they often suffer from unstable training due to their reliance on bootstrapping to propagate returns. Furthermore, the requirement to discount future rewards can induce undesirable short-sighted behaviors, limiting the agent's ability to plan over long horizons.
To overcome these issues, the DT framework \cite{NEURIPS2021_7f489f64} treats the offline RL problem as a generic conditional sequence modeling task, effectively converting it into a supervised learning problem. This enables DT to handle long sequences and avoid stability issues associated with boostrapping.

\begin{figure}[t]
\centering
\includegraphics[width=0.45\textwidth]{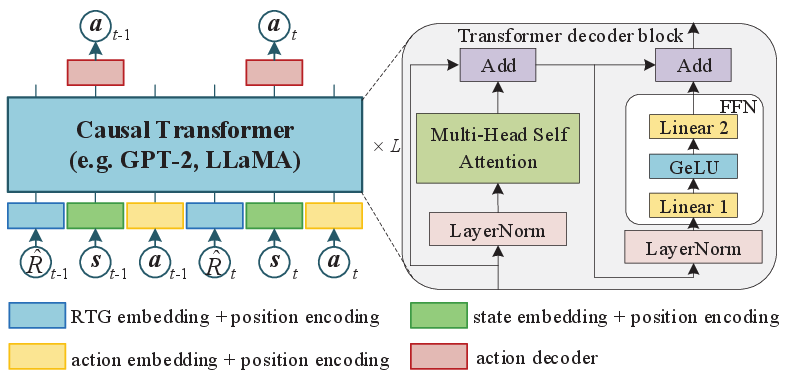}\\
\caption{The architecture of decision transformer.}
\label{fig:DT_architeture}
\end{figure}

In the offline setting, interactions with the environment are no longer available. Agents can only learn from a pre-collected trajectory dataset $\mathcal{D}=\{\bm{\tau}_{e}\}_{e=1}^E$ with $E$ trajectories generated by an unknown behavior policy $\pi_{\rm{B}}$. Each trajectory $\bm{\tau}_{e}$ ($e \in \{1,2,\cdots, E\}$) is an ordered sequence of return-to-go (RTG) values, states, and actions over the time horizon, defined as:
\begin{align}
\bm{\tau}_e = (\widehat{R}_1,\bm{s}_1,\bm{a}_1, \widehat{R}_2,\bm{s}_2,\bm{a}_2, \cdots, \widehat{R}_T,\bm{s}_T,\bm{a}_T ),
\end{align}
where $\widehat{R}_t$ is the RTG in time slot $t$, defined as the sum of future rewards from $t$ to the end of the episode, i.e.,
\begin{align}
\widehat{R}_t = \sum\nolimits_{t' = t}^T r_{t'}, \forall t \in \{1, 2, \cdots, T\}.
\end{align}

Given the offline trajectory dataset $\mathcal{D}$, we can train a DT model to learn the UAV control policy. The architecture of the DT model is illustrated in Fig. \ref{fig:DT_architeture}, which consists of three components: a data encoder module, a causal transformer module composed of a stack of transformer layers, and an action decoder module.
For each trajectory $\bm{\tau}_e$ and each time slot $t$, we feed the most recent $K$ timesteps of trajectory data into the DT model, i.e.,
\begin{align}\label{eq:sub_traj}
\bm{\tau}_{e,t} = (\widehat{R}_{t-K+1},\bm{s}_{t-K+1},\bm{a}_{t-K+1}, \cdots, \widehat{R}_{t},\bm{s}_{t},\bm{a}_{t}),
\end{align}
where $K$ is referred to as the context length. The input sequence $\bm{\tau}_{e,t}$ is first processed by the data encoder module, which comprises three linear layers and a learnable time embedding layer.
The time embedding layer encodes the time information of each step $i$ ($t-K+1 \leq i \leq t$) into a vector whose dimension matches the hidden dimension of the transformer module, i.e., $\bm{t}_i^{\rm{v}} \in \mathbb{R}^{1\times d_{\rm{trans}}}$, where $d_{\rm{trans}}$ is the hidden dimension of the transformer module.
Each of the three linear layers maps the RTG $\widehat{R}_i$, the state $\bm{s}_i$, and the action $\bm{a}_i$ at each step $i$ to vector representations, also matching the hidden dimension of the transformer.
The time embedding is then added into each representation to form the input tokens:
\begin{align}
\left\{ {\begin{array}{*{20}{c}}
\widehat{R}_{i}^{\rm{v}} = {\rm{Linear}}(\widehat{R}_{i}) + \bm{t}_i^{\rm{v}} \in \mathbb{R}^{1\times d_{\rm{trans}}}, \\
\bm{s}_{i}^{\rm{v}} = {\rm{Linear}}(\bm{s}_{i}) + \bm{t}_i^{\rm{v}} \in \mathbb{R}^{1\times d_{\rm{trans}}}, \\
\bm{a}_{i}^{\rm{v}} = {\rm{Linear}}(\bm{a}_{i}) + \bm{t}_i^{\rm{v}} \in \mathbb{R}^{1\times d_{\rm{trans}}}.
\end{array}} \right.
\end{align}

After encoding $\bm{\tau}_{e,t}$, we obtain a token sequence $\bm{\tau}_{e,t}^{\rm{v}} = (\widehat{R}_{t-K+1}^{\rm{v}},\bm{s}_{t-K+1}^{\rm{v}},\bm{a}_{t-K+1}^{\rm{v}}, \cdots, \widehat{R}_{t}^{\rm{v}},\bm{s}_{t}^{\rm{v}},\bm{a}_{t}^{\rm{v}} )$, which is then fed into the transformer module and output a corresponding sequence of predicted tokens, i.e.,
\begin{align}
&\bm{\tau}_{e,t}^{\rm{out}} = {\rm{Transformer}}(\bm{\tau}_{e,t}^{\rm{v}}) \notag\\
&= (\widehat{R}_{t-K+1}^{\rm{out}},\bm{s}_{t-K+1}^{\rm{out}},\bm{a}_{t-K+1}^{\rm{out}}, \cdots, \widehat{R}_{t}^{\rm{out}},\bm{s}_{t}^{\rm{out}},\bm{a}_{t}^{\rm{out}} ).
\end{align}
Finally, each output token $\bm{s}_i^{\rm{out}}$ ($t-K+1\le i \le t$) is passed to the action decoder, which consists of a linear layer, to predict the action for timestep $i$, i.e.,
\begin{align}
\hat{\bm{a}}_{i} = {\rm{Linear}}(\bm{s}_{i}^{\rm{out}}).
\end{align}
In summary, given a trajectory segment $\bm{\tau}_{e,t}$, the DT model outputs a sequence of predicted actions $(\hat{\bm{a}}_{t-K+1}, \ldots, \hat{\bm{a}}_{t})$, where $\hat{\bm{a}}_t$ is the predicted action at time step $t$.
The learning objective of DT model $\pi_{\bm{\theta}}$, parameterized by $\bm{\theta}$, is to predict next action $\hat{\bm{a}}_{t}$ from trajectory prefix and current state $\bm{s}_t$, using the true action $\bm{a}_t$ as the target, written as:
\begin{align}
L_{\rm{DT}}(\bm{\theta}) = \sum\limits_{i = t - K + 1}^t \left\|\bm{a}_i - \hat{\bm{a}}_i \right\|^2.
\end{align}
After training, the DT model can be deployed to interact with the environment and generate UAV control actions in an autoregressive manner.
For DT framework, we have the following theorem:
\begin{thm}\label{theorem:one}
Assume the above-defined MDP is $\varepsilon$-deterministic, i.e., $\Pr(r_t \ne \mathcal{R}(\bm{s}_t, \bm{a}_t)~{\rm{or}}~ \bm{s}_{t+1} \ne \mathcal{P}(\bm{s}_t, \bm{a}_t) \left| \bm{s}_t, \bm{a}_t) \right.) \le \varepsilon$  at all $\bm{s}_t$, $ \bm{a}_t$ for the reward function $\mathcal{R}$ and environment dynamics $\mathcal{P}$. Let $g(\bm{\tau})=\sum\nolimits_{t = 1}^T r_t$. Given a dataset $\mathcal{D}$ pre-collected by behavior policy $\pi_{\rm{B}}$, which contains a sufficient number of trajectories whose returns match the DT's conditioning values $\hat{R}_1$, i.e, $\Pr(g(\bm{\tau}) = \hat{R}_1\left| \bm{s}_1, \bm{\tau} \in \mathcal{D} \right.) \ge \zeta$ for all initial states $\bm{s}_1$. Then, the DT policy $\pi_{\bm{\theta}}$ learned from this dataset satisfy:
\begin{align}
\mathbb{E}_{\bm{\tau} \sim \pi_{\rm{B}}}\left[ g(\bm{\tau}) \right] - \mathbb{E}_{\bm{\tau} \sim \pi_{\bm{\theta}}}\left[ g(\bm{\tau}) \right] \le \varepsilon \Big(\frac{1}{\zeta} + 2 \Big) T^2
\end{align}
\end{thm}
\begin{proof}
See Appendix \ref{app:one}.
\end{proof}

According to Theorem \ref{theorem:one}, training a DT model with the loss function $L_{\rm{DT}}$ drives the learned policy to gradually converge toward the behavior policy $\pi_{\rm{B}}$. However, this convergence inherently restricts the learned DT policy from surpassing the performance of the underlying behavior policy $\pi_{\rm{B}}$ that used to generate the trajectories in the offline dataset $\mathcal{D}$. Furthermore, training exclusively with the $L_{\rm{DT}}$ limits the DT's ability to stitch together high-reward actions across different trajectories, leaving the learned DT policy $\pi_{\bm{\theta}}$ largely biased toward actions observed during training.

\section{LLM-Empowered Critic-Regularized Decision Transformer for UAV Control}\label{sec:LLM_CRDT}
To address the limitations of DT approaches discussed above, this section presents an LLM-empowered critic-regularized DT framework, i.e., LLM-CRDT, to solve problem \eqref{prob:P_UAV} and facilitate efficient online UAV control. Specifically, we first propose a critic-regularized DT learning paradigm that endows the DT model with stitching capabilities. Then, we employ a pre-trained LLM as the transformer backbone of the DT model and employ LoRA for fine-tuning, thereby exploiting the LLM's few-shot generalization ability to further enhance UAV control performance.

\begin{figure}[ht]
\centering
\subfigure[]{\includegraphics[width=0.49\textwidth]{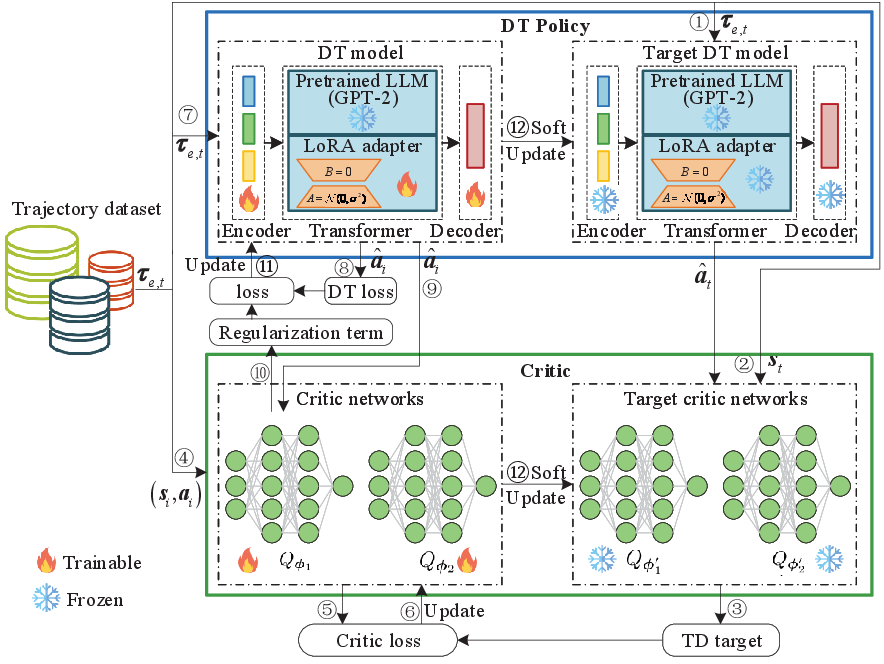}\label{fig:LLMCRDT_finetune}}\\
\subfigure[]{\includegraphics[width=0.48\textwidth]{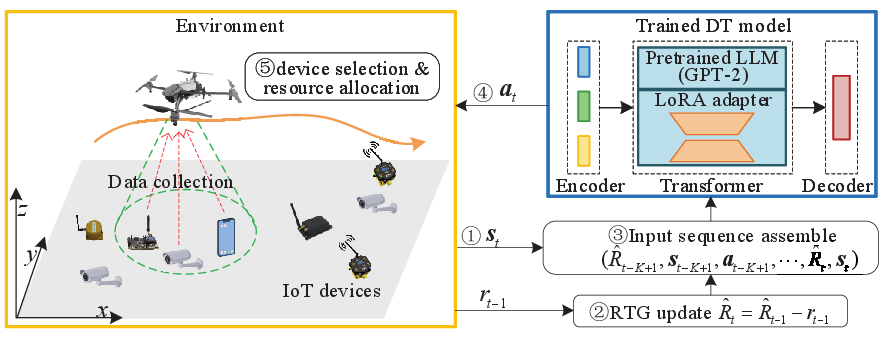}\label{fig:LLMCRDT_inference}}
\caption{The architecture of the proposed LLM-CRDT: (a) Fine-tuning on an offline trajectory dataset, (b) Inference to generate action and interact with environment.}
\label{fig:LLMCRDT}
\end{figure}

\subsection{Critic-Regularized Decision Transformer}
Inspired by the stitching capability inherent in actor-critic methods \cite{sutton1998reinforcement}, we propose a critic-regularized DT learning paradigm, which integrates the sequence modeling capacities of DT with the policy improvement guidance provided by critic networks, as shown in Fig. \ref{fig:LLMCRDT_finetune}. Unlike the standard DT approach in \cite{NEURIPS2021_7f489f64} that rely solely on supervised learning over return-conditioned trajectories, our approach introduces a state-action value-based regularization term into the training loss. This encourages the model to select actions that align with high value predictions from the critic, thereby improving policy performance beyond that of the behavior trajectories in the offline dataset.

Specifically, we draw inspiration from the actor-critic algorithms (such as SAC \cite{pmlr-v80-haarnoja18b} and TD3 \cite{pmlr-v80-fujimoto18a}) and employ two critic networks, $Q_{\bm{\phi}_1}$ and $Q_{\bm{\phi}_2}$, along with their corresponding target networks, $Q_{\bm{\phi}_1'}$ and $Q_{\bm{\phi}_2'}$, to learn accurate estimates of the state-action value function.
In addition, we employ a DT model $\pi_{\bm{\theta}}$, parameterized by $\bm{\theta}$, to predict actions.
Following the common practice in conventional RL approaches, we maintain a target DT model $\pi_{\bm{\theta}'}$ to enhance the stability of the learning process.
Considering that the DT policy conditions on past trajectory information, we estimate the state-action values via the $n$-step Bellman backup, which has been shown to outperform the 1-step approximation \cite{sutton1998reinforcement}. Given a trajectory segment $\bm{\tau}_{e,t}$ as defined in \eqref{eq:sub_traj}, we use the target DT policy $\pi_{\bm{\theta}'}$ to predict the action at time step $t$, i.e., $\hat{\bm{a}}_t$. We then compute the TD target for the $i$-th step ($t-K+1\le i \le t-1$) as follows:
\begin{align}\label{eq:TD_target}
\hat{Q}_i = \sum\limits_{j = i}^{t - 1} \gamma^{j - i} r_j + \gamma ^{t - i} \mathop {\min}\limits_{h \in \{1,2\}} Q_{\bm{\phi}_h'}(\bm{s}_t,\hat{\bm{a}}_t).
\end{align}
Next, we use the critic networks $Q_{\bm{\phi}_1}$ and $Q_{\bm{\phi}_2}$ to estimate the state-action values at each time step $i$, and their parameters are optimized by minimizing the following loss function:
\begin{align}\label{eq:critic_loss}
L(\bm{\phi}_h) = \sum\limits_{i = t - K + 1}^{t - 1} {\Big\| \hat{Q}_i - Q_{\bm{\phi}_h}(\bm{s}_i, \bm{a}_i) \Big\|^2}, h \in \{1, 2\}.
\end{align}

During the DT model training, we incorporate a state-action value loss term to regularize the DT towards sampling high-value actions. Specifically, we first use the DT model $\pi_{\bm{\theta}}$ to predict the actions in each time step $i \in \{t-K+1, \cdots, t\}$, i.e., ($\hat{\bm{a}}_{t-K+1}, \cdots, \hat{\bm{a}}_t$). Then, we use the critic network $Q_{\bm{\phi}_1}$ to evaluate the state-action values for these predicted actions, i.e, $\{Q_{\bm{\phi}_1}(\bm{s}_i, \hat{\bm{a}}_i): \forall i \in \{t-K+1, \cdots, t\}\}$. Hence, the DT model is optimized by minimizing the following loss function:
\begin{align}\label{eq:loss_CRDT}
L(\bm{\theta}) &= L_{\rm{DT}}(\bm{\theta}) + \lambda \cdot L_Q(\bm{\theta}) \notag\\
& = L_{\rm{DT}}(\bm{\theta}) - \lambda \sum \limits_{i=t-K+1}^t Q_{\bm{\phi}_1}(\bm{s}_i, \hat{\bm{a}}_i),
\end{align}
where $L_Q(\bm{\theta})= -\sum \nolimits_{i=t-K+1}^t Q_{\bm{\phi}_1}(\bm{s}_i, \hat{\bm{a}}_i)$ is the regularization term to evaluate the DT policy and promotes high-reward actions, $\lambda > 0$ is a hyperparameter used to balance the trajectory modeling loss and the regularization loss.
By training with the loss function \eqref{eq:loss_CRDT}, the DT model acquires stitching capabilities and is able to outperform the underlying behavior policy used to collect the dataset.

In addition, to stabilize learning process, we adopt the soft-update strategy \cite{sutton1998reinforcement} to update the target critic networks ($\phi_1'$ and $\phi_2'$) and the target policy network $\bm{\theta}'$. With mixing coefficient $\rho \in (0, 1)$, the target networks are updated as
\begin{align}\label{eq:target_network_update}
\left\{ {\begin{array}{*{20}{c}}
\bm{\theta}' = \rho \bm{\theta} + (1-\rho) \bm{\theta}',\\
\bm{\phi}_h' = \rho \bm{\phi}_h + (1-\rho) \bm{\phi}_h', h \in \{1, 2\}.
\end{array}} \right.
\end{align}

\subsection{LLM-Empowered Critic-Regularized Decision Transformer}

\begin{algorithm}[t] \small
\caption{Fine-tuning of LLM-CRDT}
\label{alg:finetuning}
\begin{spacing}{0.85}
\begin{algorithmic}[1]
\State \textbf{Input}: Offline trajectory dataset $\mathcal{D}$, context length $K$, soft update coefficient $\rho$, LoRA rank $r$, scaling factor $\alpha$, batch size $B$.
\State Create LoRA adapters for the query, key, value projection layers and attention output layer in each transformer block of the LLM.
\State Freeze the LLM backbone, enable the parameters of data encoder, LoRA adapters, action decoder, and critics to be trainable.
\Repeat
\State \parbox[t]{\dimexpr\linewidth-\algorithmicindent}{Randomly sample a batch of $B$ trajectories $\mathcal{B}=\{\bm{\tau}_e\}_{e=1}^B$ from $\mathcal{D}$.}
\State \parbox[t]{\dimexpr\linewidth-\algorithmicindent}{For each trajectory $\bm{\tau}_e \in \mathcal{B}$, randomly select a time slot $t_e$ and extract the corresponding $K$-step segment $\bm{\tau}_{e, t_e}$, forming a batch of trajectory segments, i.e., $\mathcal{B}_s = \{\bm{\tau}_{e,t_e}\}_{e=1}^B$.}
\State \parbox[t]{\dimexpr\linewidth-\algorithmicindent}{Pass $\mathcal{B}_s$ through the target DT $\pi_{\bm{\theta}'}$ to predict actions for time slot $t_e$, i.e., $\{\hat{\bm{a}}_{e,t_e}: e \in \{1, \cdots, B\}\}$.}
\State \parbox[t]{\dimexpr\linewidth-\algorithmicindent}{Compute the TD target for every time slot of each trajectory segment $\bm{\tau}_{e, t_e}$ based on \eqref{eq:TD_target}.}
\State \parbox[t]{\dimexpr\linewidth-\algorithmicindent}{Compute the critic losses for $\bm{\phi}_1$ and $\bm{\phi}_2$ based on \eqref{eq:critic_loss}, and perform gradient descent to optimize them.}
\State \parbox[t]{\dimexpr\linewidth-\algorithmicindent}{Input the $\mathcal{B}_s$ into the DT model $\pi_{\bm{\theta}}$ to compute the predicted actions, i.e., $\hat{\bm{a}} = \{\hat{\bm{a}}_{h,r}: 1\le h \le B, t_h - K + 1 \le r \le t_h\}$.}
\State \parbox[t]{\dimexpr\linewidth-\algorithmicindent}{Compute the DT loss based on \eqref{eq:loss_CRDT}, and optimize it via gradient descent.}
\State \parbox[t]{\dimexpr\linewidth-\algorithmicindent}{update the target critic networks and target actor network based on \eqref{eq:target_network_update}.}
\Until{The model is converged.}
\end{algorithmic}
\end{spacing}
\end{algorithm}

Up to now, the proposed critic-regularized DT gains stitching capability. However, transformers are known to be data-hungry and usually require massive amounts of training data to achieve satisfactory performance. As a result, training the DT model from scratch can be both challenging and computationally expensive. To this end, this subsection proposes LLM-CRDT to leverage the few-shot generalization capabilities of pre-trained LLMs to enhance the DT model for UAV control.

Similar to the standard DT framework in \cite{NEURIPS2021_7f489f64}, we adopt GPT-2 \cite{radford2019language} as the transformer module in our proposed LLM-CRDT. However, unlike prior DT works (e.g., \cite{NEURIPS2021_7f489f64, 10848143}) that randomly initialize model parameters and train DT models from scratch, we initialize the transformer with pre-trained GPT-2 parameters. This enables the DT model to exploit the general knowledge encoded in LLMs, thereby enhancing action prediction. Note that the transformer module in DT can be flexibly replaced with other casual LLMs, such as LLaMA \cite{touvron2023llama}, with the choice of LLM depending on the tradeoff between computational cost and performance requirements.

After initialization, we fine-tune LLM-CRDT, which consists of the data encoder, the pre-trained LLM, the action decoder, and the critic networks. Notably, the pre-trained LLM accounts for the majority of the model parameters, whereas the parameter sizes of the other modules are comparatively negligible. Given the large scale of LLMs, full-parameter fine-tuning is computationally expensive and risks overwriting the valuable knowledge embedded in the pre-trained LLM, potentially diminishing its few-shot generalization capabilities. To this end, we adopt a parameter-efficient fine-tuning approach, i.e., LoRA \cite{hu2022lora}, to efficiently adapt the pre-trained LLM for UAV control, thereby substantially reducing the number of trainable parameters. Specifically, we apply LoRA to the query, key, and value projection layers, as well as the self-attention output layer in each transformer block of the LLM. For a given layer with parameter matrix $\mathbf{W} \in \mathbb{R}^{d_{\rm{LLM}} \times d_{\rm{LLM}}}$, where $d_{\rm{LLM}}$ denotes the LLM's hidden dimension. LoRA injects trainable low-rank matrices $\mathbf{A} \in \mathbb{R}^{d_{\rm{LLM}} \times r}$ and $\mathbf{B} \in \mathbb{R}^{r \times d_{\rm{LLM}}}$ into the layer. The original weight matrix $\mathbf{W}$ is then modified as
\begin{align}
\mathbf{W} \leftarrow \mathbf{W} + \frac{\alpha}{r} \mathbf{B}\mathbf{A},
\end{align}
where $r$ is the LoRA adapter rank and $\alpha$ controls its influence. During fine-tuning, the weight matrix $\mathbf{W}$ remains frozen, while only the LoRA adapter matrices $\mathbf{A}$ and $\mathbf{B}$ are updated. Typically, $\mathbf{A}$ is initialized with Gaussian distribution values, and $\mathbf{B}$ is initialized to zeros. Since $r \ll d_{\rm{LLM}}$, the number of trainable parameters introduced by $\mathbf{A}$ and $\mathbf{B}$, i.e., $2rd_{\rm{LLM}}$, is significantly smaller than that of $\mathbf{W}$ (i.e., $d_{\rm{LLM}}^2$). Consequently, LoRA enables efficient and cost-effective fine-tuning. Furthermore, by keeping the original LLM weights fixed, LoRA mitigates catastrophic forgetting, thereby enabling us to retain and exploit the universal reasoning and generalization capabilities of the pre-trained LLM for UAV control.

With the assistance of LoRA, we fine-tune the LLM-CRDT to learn the UAV control policy from an offline trajectory dataset $\mathcal{D}=\{\bm{\tau}_e\}_{e=1}^{E}$ with $E$ trajectories. In particular, the LLM weights are frozen, while only the parameters of the data encoder, LoRA adapters, action decoder, and critic networks are updated, as shown in Fig. \ref{fig:LLMCRDT_finetune}. In each fine-tuning epoch, we randomly sample a batch of $B$ trajectories $\mathcal{B}=\{\bm{\tau}_e\}_{e=1}^B$. For each trajectory $\bm{\tau}_e \in \mathcal{B}$, we randomly select a time slot $t_e$ and extract the corresponding $K$-step segment $\bm{\tau}_{e, t_e}$ ending at $t_e$, thereby forming a batch of trajectory segments $\mathcal{B}_s = \{\bm{\tau}_{e,t_e}\}_{e=1}^B$. We then pass $\mathcal{B}_s$ through the target DT model to predict the action at time step $t_e$ for each trajectory segment $\bm{\tau}_{e,t_e}$, yielding $\hat{\bm{a}}_{e, t_e}$. Following that, we compute TD targets for each time step based on \eqref{eq:TD_target}, evaluate the critic losses via \eqref{eq:critic_loss}, and optimize the critic networks via gradient descent. Subsequently, we feed $\mathcal{B}_s$ into the DT model to obtain the predicted actions for each time step of each trajectory segment, i.e., $\hat{\bm{a}}_{e,i}$ for $e = 1, \dots, B$ and $i = t_e - K + 1, \dots, t_e$.
Finally, we compute the loss of the DT model based on \eqref{eq:loss_CRDT} and perform gradient descent to update the parameters of DT model. For clarity, the detailed steps for fine-tuning the proposed LLM-CRDT model are summarized in Algorithm \ref{alg:finetuning}.

\subsection{Inference of LLM empowered Decision Transformer}
After fine-tuning, we deploy the proposed LLM-CRDT model for online inference to control the UAV and complete the data collection mission. In the deployment phase, only the trained DT model $\pi_{\bm{\theta}}$ is required, while the target DT model and critic networks are utilized solely during fine-tuning and are not involved in inference. During deployment, $\pi_{\bm{\theta}}$ interacts with the environment in an autoregressive loop using a context window of length $K$, consistent with the context length adopted in fine-tuning phase, as illustrated in Fig. \ref{fig:LLMCRDT_inference}.

\begin{algorithm}[t] \small
\caption{Inference of LLM-CRDT for UAV control}
\label{alg:inference}
\begin{spacing}{0.85}
\begin{algorithmic}[1]
\State \textbf{Input}: Trained DT model $\pi_{\bm{\theta}}$, target RTG $\widehat{R}_1$, context length $K$.
\State Initialize the environment and obtain the initial state $\bm{s}_1$.
\State Initialize the trajectory history as $\bm{\tau}=\{(\bm{s}_1, \hat{R}_1)\}$ and time slot $t=1$.
\For{$t = 1:T$}
    \State \parbox[t]{\dimexpr\linewidth-\algorithmicindent}{Construct the input sequence $\bm{\tau}_t$ by concatenating the $K$ most-recent completed triples $(\widehat{R}_i, \bm{s}_i, \bm{a}_i)$ for $i= t-K, \cdots, t-1$ in $\bm{\tau}$, followed by the current pair $(\widehat{R}_t, \bm{s}_t)$.}
    \State Left-pad the input sequence $\bm{\tau}_t$ if $t<K$.
    \State Input $\bm{\tau}_t$ into the DT model and obtain predicted action $\hat{\bm{a}}_t$.
    \State \parbox[t]{\dimexpr\linewidth-\algorithmicindent}{Execute action $\bm{a}_t$ in the environment, apply Algorithm \ref{alg:RB_allocation} for device selection and resource allocation, and then obtain the next state $\bm{s}_{t+1}$ and reward $r_t$.}
    \State Update the RTG as $\hat{R}_{t+1} = \hat{R}_t - r_t$
    \State Append the triplets $(\widehat{R}_{t}, \bm{s}_{t}, \hat{\bm{a}}_t)$ to trajectory $\bm{\tau}$.
    \If{termination condition is met}
        \State Break the circulation.
    \EndIf
\EndFor
\end{algorithmic}
\end{spacing}
\end{algorithm}

At the beginning of an episode, we initialize the environment to obtain the initial state $\bm{s}_1$ and specify a desired RTG $\widehat{R}_1$, which is usually set as the maximum possible return in the training dataset \cite{NEURIPS2021_7f489f64}. Then, we perform the DT model $\pi_{\bm{\theta}}$ to generate UAV control actions in each slot to complete the data collection task.
Specifically, in each time slot $t$, given the current state $\bm{s}_t$ and RTG $\widehat{R}_t$, we assemble the model's input sequence as the most recent $K$ triplets $(\widehat{R}_i, \bm{s}_i, \bm{a}_i)$ up to time slot $t-1$, followed by the pair $(\widehat{R}_t, \bm{s}_t)$.
If fewer than $K$ time slots are available, i.e., $t<K$, the sequence is left-padded and masked.
The input sequence is then processed in the same manner as during fine-tuning: it is first encoded by the data encoder, passed through the LLM, and finally decoded by the action decoder to yield the predicted action $\hat{\bm{a}}_t$. The action is executed in the environment, producing the next state $\bm{s}_{t+1}$ and reward $r_t$. The RTG is subsequently updated as
\begin{align}
\hat{R}_{t+1}=\hat{R}_t-r_t.
\end{align}
Furthermore, the new triplet $(\hat{R}_t,\bm{s}_t,\hat{\bm{a}}_t)$ is appended to the trajectory history and the context window is advanced to retain at most $K$ steps.
This process is repeated until the episode terminates.
For clarity, we summarize the detailed steps of the LLM-CRDT inference for UAV control in Algorithm \ref{alg:inference}.

\section{Simulation Results}\label{sec:simulations}
This section evaluates the effectiveness of the proposed LLM-CRDT and resource allocation approach. In the simulations, we consider the devices are randomly distributed within a rectangular area with side length $X_{\max}=1$ km and $Y_{\max}=1$ km. The initial data volume of each device $n$, i.e., $D_n$, is randomly selected from $[0.5 D_{\max}, D_{\max}]$. Unless specified otherwise, the default settings of the UAV-enabled IoT network are summarized in Table \ref{tab:simu_settings}, which follow widely adopted configurations for UAV-enabled IoT networks in prior works, e.g., \cite{10044099, 9826413, 9562293, 10706837}. In the proposed LLM-CRDT model, GPT-2 serves as the LLM backbone, while both critic networks adopt a three-layer MLP architecture with Mish activations and 512 hidden units per layer. For dataset preparation, a behavior policy is first trained using the SAC algorithm, after which an offline dataset of 1000 trajectories is collected by rolling out this policy in the environment. The fine-tuning of LLM-CRDT is performed over 1000 epochs with a context length of $K=20$, a LoRA rank of $r=16$, a LoRA scaling factor of $\alpha=32$, a batch size of $B=128$, a soft update coefficient of $\rho=0.005$, and a loss coefficient of $\lambda=1.0$. Both the DT model and the critic networks are optimized using the AdamW optimizer with a learning rate of $10^{-5}$.

\begin{table}[ht]\small
\caption{System Parameters}
\label{tab:simu_settings}
\centering
\begin{tabular}{p{1.5cm}|p{1.3cm}|p{1.5cm}|p{1.3cm}}
\hline
Parameter & Value & Parameter & Value \\
\hline
$N$ & 100  & $H$ & 100 m \\
$V_{\max}$ & 25 m/s & $\iota$ & 2.2 \\
$\kappa$ & 0.2 & $a$ & 15 \\
$b$ & 0.5 & $g_0$ & -50 dB \\
$p_n$ & 0.1 W & $D_{\max}$ & 20 MB \\
$M$& 10 & $W$ & 1 MHz \\
$N_0$ & -174 dBm & $\omega_{\max}$ & 1.0 \\
$T$ & 500 & $\delta$ & 5 s \\
$P_1$ & 79.86 W & $P_2$ & 88.63 W \\
$U_{\rm{tip}}$ & 120 m/s & $v_0$ & 4.03 m /s \\
$d_0$ & 0.6 & $\rho$ & 1.225 \\
$g$ & 0.05 & $A$ & 0.503 $\rm{m}^2$ \\
\hline
\end{tabular}
\vspace{-0.5cm}
\end{table}

\subsection{Effectiveness of LLM-CRDT}\label{subsec:simu_A}
In this subsection, we verify the effectiveness of the proposed LLM-CRDT by comparing it with the following online and offline RL approaches for UAV trajectory planning.
Note that, for all benchmark approaches, device selection and resource allocation decisions are obtained using our proposed Algorithm \ref{alg:RB_allocation}, while the UAV control policies are determined by each respective method.
\begin{itemize}
  \item TD3 \cite{pmlr-v80-fujimoto18a}: TD3 is an online RL method that combats overestimation by using two critic networks and stabilizes training via delayed policy updates.
  \item SAC \cite{pmlr-v80-haarnoja18b}: SAC is a online RL algorithm that combines off-policy learning with entropy maximization, encouraging exploration and improving robustness.
  \item TD3+BC \cite{NEURIPS2021_a8166da0}: TD3+BC is an offline reinforcement learning approach that modifies the policy update of the TD3 by incorporating a behavior cloning term to regularize the policy, enabling effective learning in offline settings.
  \item SAC-N \cite{NEURIPS2021_3d3d286a}: SAC-N is an offline RL approach that modifies the SAC algorithm by using an ensemble of multiple critic networks to enhance performance in offline settings. In this work, we employ four critic networks for the SAC-N approach.
  \item Behavior Cloning (BC) \cite{sutton1998reinforcement}: BC is an imitation learning method which learns a policy via supervised learning to mimic actions in the training dataset.
  \item DT \cite{NEURIPS2021_7f489f64}: DT is an offline RL approach that reframes policy learning as a sequence modeling problem, leveraging a transformer architecture to predict actions from past states, actions, and desired returns.
  \item LLM-DT: This approach uses pretrained LLMs to initialize the transformer module in DT and applies LoRA for fine-tuning. The only difference from the proposed LLM-CRDT is that the critic networks are not involved for regularization during the fine-tuning process.
\end{itemize}

\begin{figure}[ht]
\centering
\subfigure[]{\label{fig:offline_loss}
\includegraphics[width=0.21\textwidth]{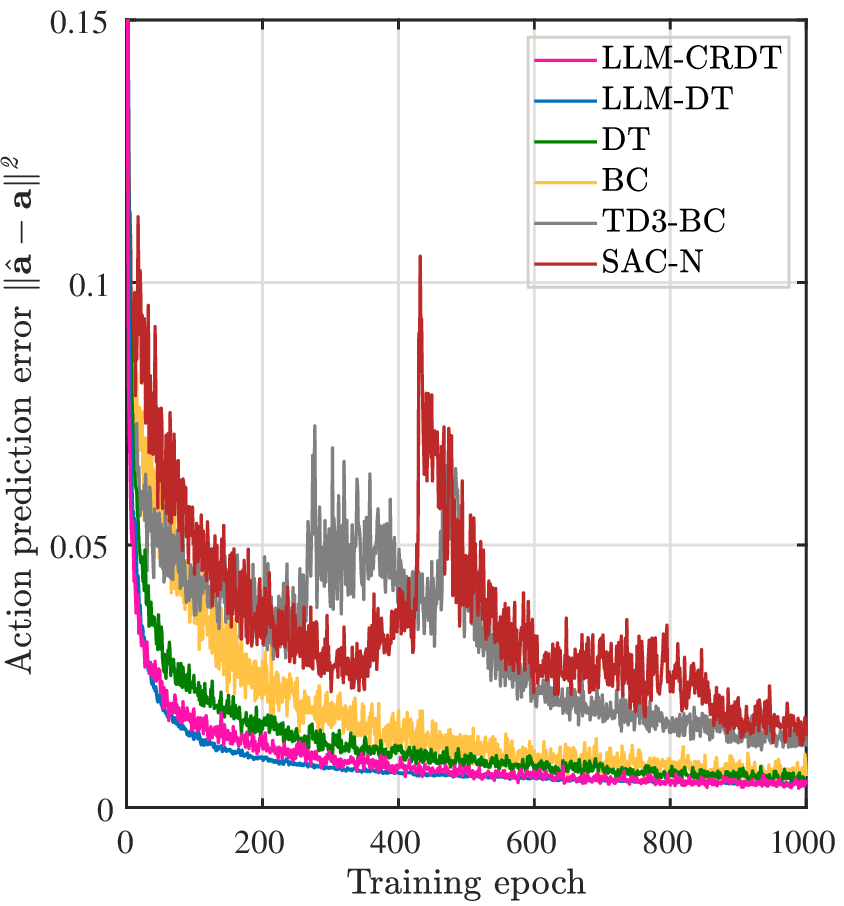}}
\subfigure[]{\label{fig:offline_reward}
\includegraphics[width=0.21\textwidth]{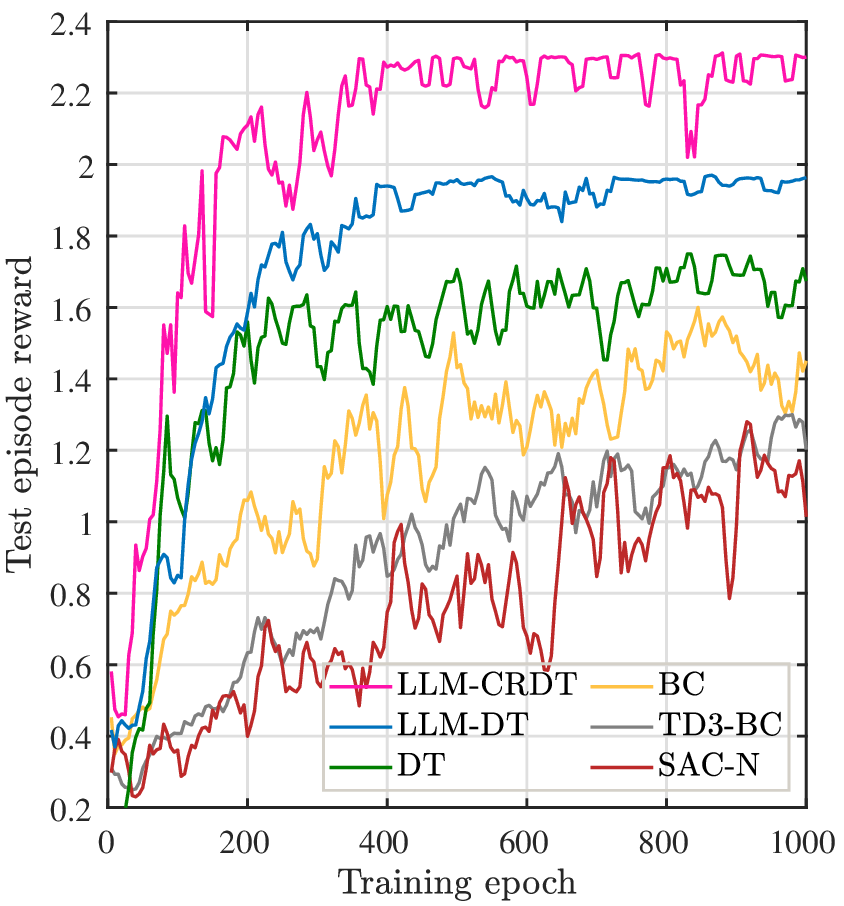}}
\caption{Comparison of learning performance for offline RL approaches.}
\label{fig:offlineAlgs}
\end{figure}

\begin{figure}[ht]
\centering
\subfigure[]{\label{fig:EE_Numdevice}
\includegraphics[width=0.46\linewidth]{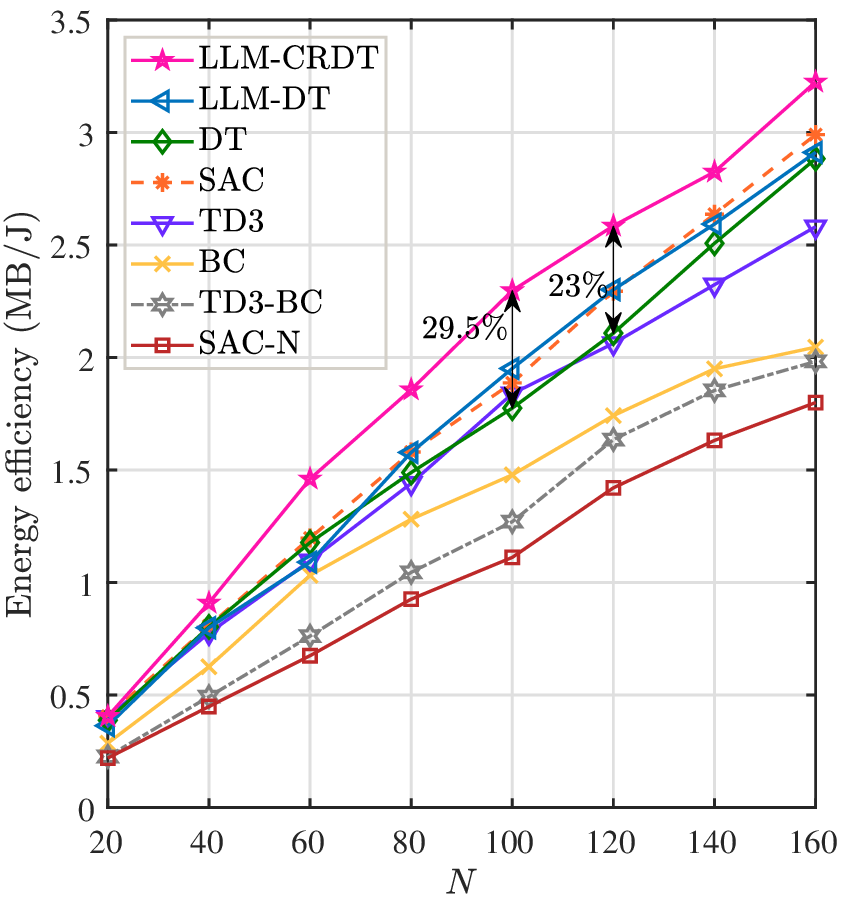}}
\subfigure[]{\label{fig:EE_Dmax}
\includegraphics[width=0.46\linewidth]{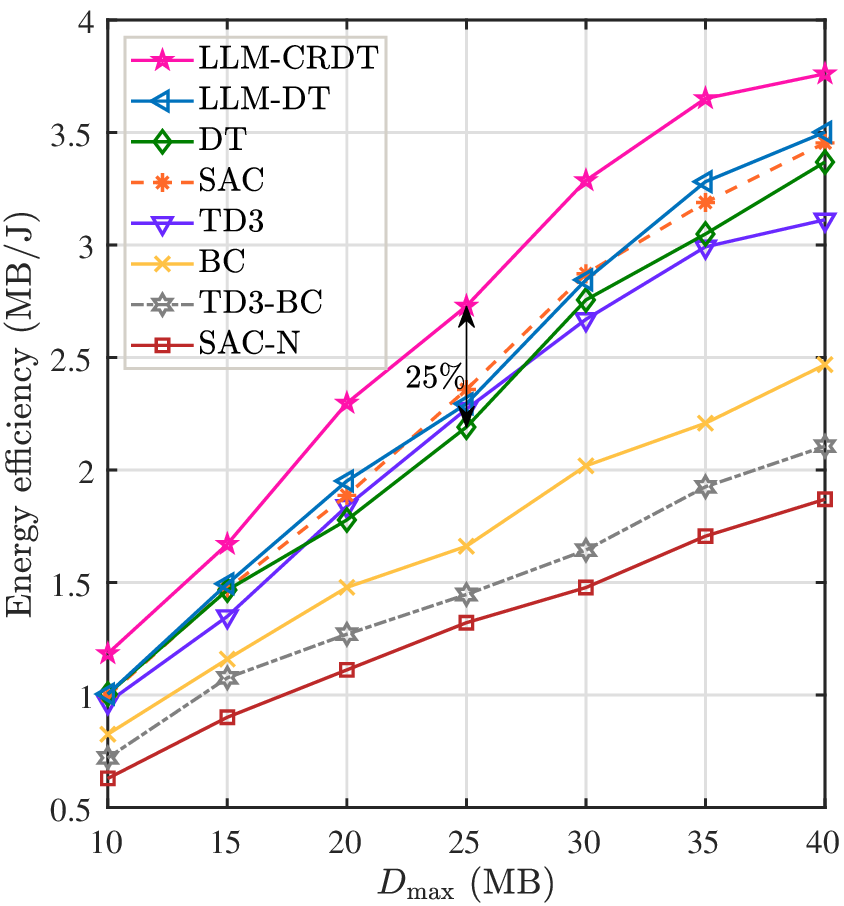}}\\
\subfigure[]{\label{fig:EE_Angle}
\includegraphics[width=0.46\linewidth]{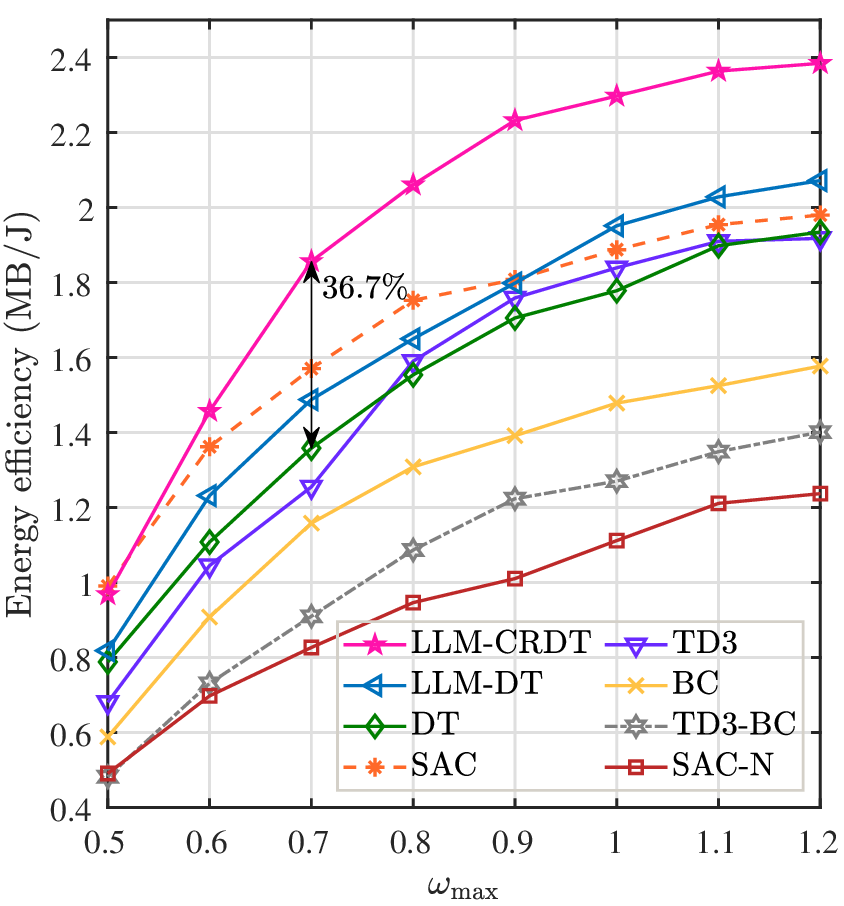}}
\subfigure[]{\label{fig:EE_height}
\includegraphics[width=0.46\linewidth]{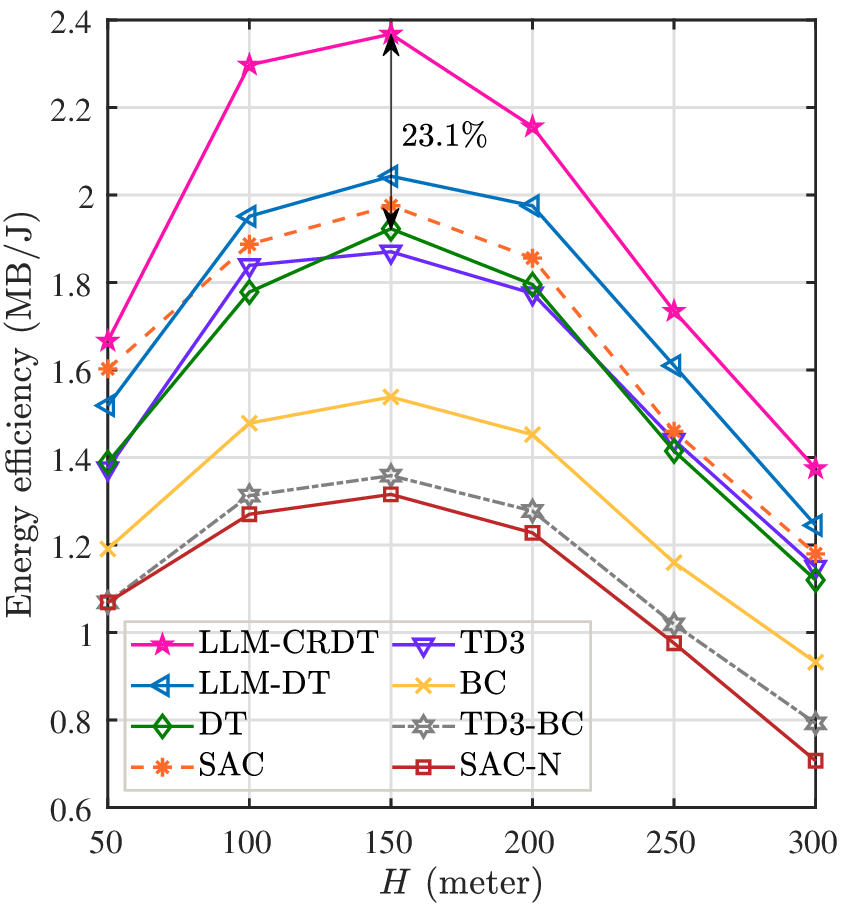}}
\caption{Comparison between the proposed LLM-CRDT and benchmarks on different system settings: (a) The number of IoT devices $N$; (b) The maximum data amount of devices $D_{\max}$; (c) The maximum azimuth angle of the UAV $\omega_{\max}$; (d) The flight height of the UAV $H$. }
\label{fig:sys_settings_EE}
\end{figure}

\begin{figure*}[ht]
\centering
\subfigure[]{\label{fig:Data_behavior}
\includegraphics[width=0.28\linewidth]{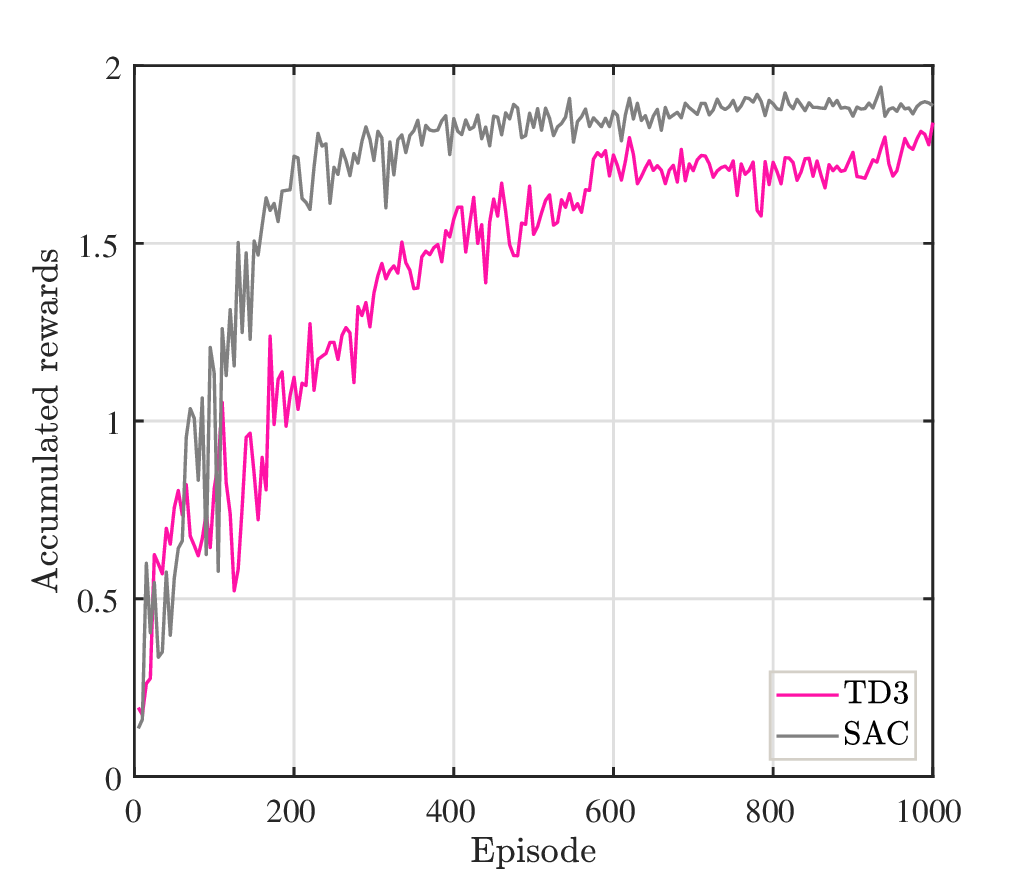}}
\subfigure[]{\label{fig:Data_quality}
\includegraphics[width=0.28\linewidth]{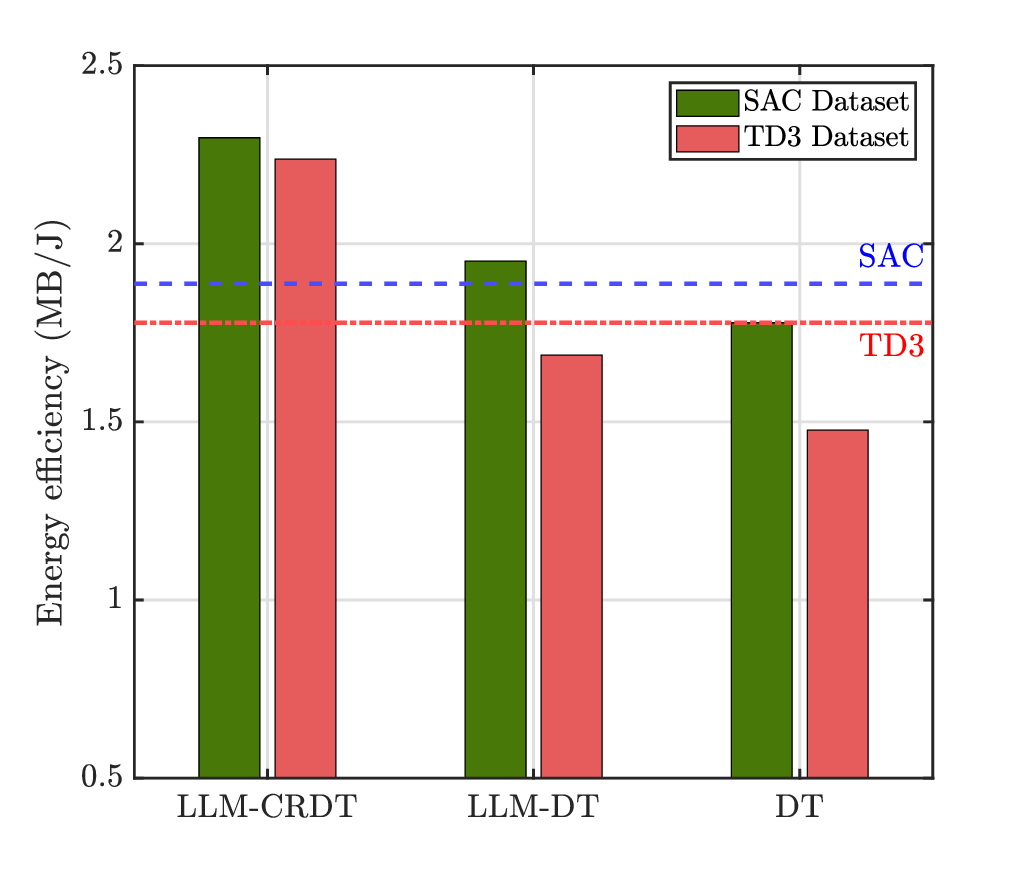}}
\subfigure[]{\label{fig:Data_amount}
\includegraphics[width=0.28\linewidth]{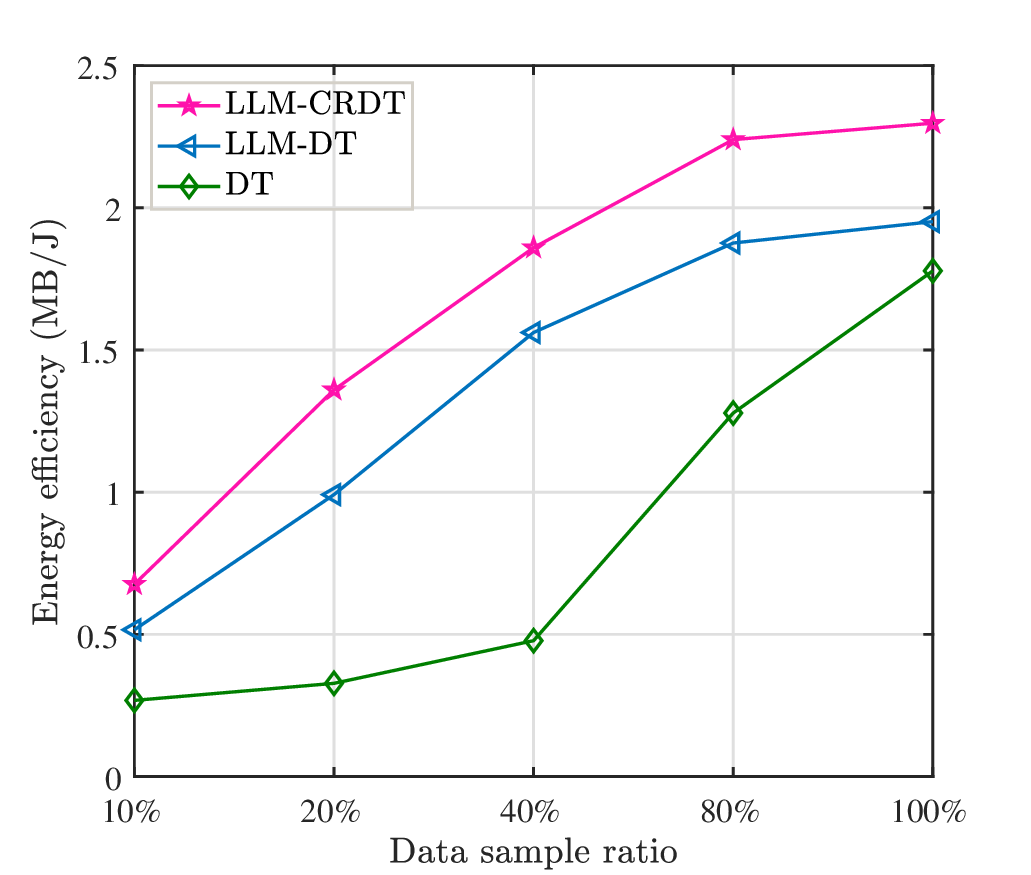}}
\caption{Impacts of training dataset quality and quantity on the proposed LLM-CRDT and DT approaches: (a) Learning performance of the behavior policy used to generate the offline dataset; (b) Effect of dataset quality; (c) Effects of training data quantity.}
\label{fig:Data_QA}
\end{figure*}

Fig. \ref{fig:offlineAlgs} compares the learning performance of our proposed LLM-CRDT to benchmark offline RL approaches. Fig. \ref{fig:offline_loss} illustrates the action prediction errors of all offline approaches, which gradually decrease and eventually converge as training progresses. This indicates that all methods are able to fit the actions in the dataset and ultimately align with the underlying behavior policy. Fig. \ref{fig:offline_reward} presents the episode reward, where the proposed LLM-CRDT outperforms all benchmark approaches. This demonstrates the effectiveness of the proposed critic-regularized learning paradigm and the fine-tuning of pre-trained LLMs in exploiting their universal reasoning capabilities. In addition, although BC, TD3-BC, and SAC-N are able to fit the behavior policy of the training dataset, their evaluation rewards are unstable during training and remain significantly lower than those of DT and our proposed LLM-CRDT. This performance gap arises because these conventional offline RL methods are highly susceptible to out-of-distribution states and actions not present in the training dataset. In essence, these methods perform simple imitation of the behavior policy without effective generalization beyond the observed trajectories.

Fig. \ref{fig:sys_settings_EE} compares the proposed LLM-CRDT with benchmark online and offline RL approaches. Specifically, Fig. \ref{fig:EE_Numdevice} shows energy efficiency of all approaches increasing with device number, since higher device density enables the UAV to collect more data efficiently. Compared with DT, our LLM-CRDT achieves a 29.5\% improvement in energy efficiency.
Fig. \ref{fig:EE_Dmax} and Fig. \ref{fig:EE_Angle} compare performance under varying maximum data amounts and maximum azimuth angles, showing that LLM-CRDT achieves up to 25\% and 36.7\% higher energy efficiency than the benchmarks, respectively. Fig. \ref{fig:EE_height} shows that the performance of all methods grows with height at first, then deteriorates as flying height increases further. This is because a higher altitude expands the UAV's sensing range, whereas excessive height degrades transmission performance between devices and the UAV. Moreover, the proposed approach surpasses DT by 25\%. It is also observed that the performance of DT nearly matches that of SAC, which is used to collect the training dataset. This observation verifies the theoretical results in Theorem \ref{theorem:one}, namely that DT is constrained by the quality of its training dataset and cannot significantly outperform the underlying behavior policy. In contrast, LLM-DT achieves noticeable improvements over DT, attributed to the universal knowledge embedded in the pre-trained LLM that enhances action prediction. Finally, LLM-CRDT consistently outperforms LLM-DT, demonstrating the effectiveness of critic regularization in endowing the DT model with stitching capabilities for superior decision-making.

\subsection{Effects of Offline Dataset Quality}
In this subsection, we assess how offline data quality and quantity influence LLM-CRDT's performance. For comparison, in addition to the SAC-generated dataset used in the previous simulations, we employ a trained TD3 agent to generate another dataset by rolling it out in the environment, each containing 1000 trajectories. For clarity, we hereafter refer to these two datasets as the SAC dataset and the TD3 dataset, respectively. To illustrate their quality, Fig. \ref{fig:Data_behavior} presents the learning performance of the underlying agents, i.e., SAC and TD3. As observed, SAC outperforms TD3, indicating that the SAC dataset is of higher quality than the TD3 dataset.

Fig. \ref{fig:Data_quality} reports how training-data quality affects the proposed LLM-CRDT and the baseline methods. For both the SAC dataset and TD3 dataset, it is observed that the performance of these approaches consistently follows the order: LLM-CRDT $>$ LLM-DT $>$ DT. An interesting phenomenon is that LLM-CRDT trained on the TD3 dataset performs almost identically to LLM-CRDT trained on the SAC dataset. This demonstrates that the proposed critic-regularized learning paradigm equips LLM-CRDT with a stitching ability, allowing it to combine actions from suboptimal trajectories into better decisions. Moreover, LLM-CRDT trained on the TD3 dataset still outperforms both LLM-DT and DT trained on the SAC dataset, further confirming its capability to learn effective UAV control policies from suboptimal dataset. Therefore, the proposed LLM-CRDT is able to address the key limitation of DT approaches by mitigating the reliance of expert-level dataset.

Fig. \ref{fig:Data_amount} evaluates the effect of training data quantity on the performance of the proposed LLM-CRDT and the benchmark approaches. Specifically, all methods are trained using 10\%, 20\%, 40\%, 80\%, and 100\% of the SAC dataset. The results show that LLM-CRDT consistently outperforms both LLM-DT and DT across all settings. Notably, in the low-data regime (i.e., below 40\%), the DT model shows little performance improvement, whereas the proposed LLM-CRDT and LLM-DT still achieve satisfactory results. This demonstrates that leveraging a pre-trained LLM effectively reduces the reliance on large-scale training data by exploiting the universal knowledge embedded in the model.

\subsection{Effectiveness of device scheduling}
This subsection evaluates the proposed device selection and resource allocation approach, i.e., Algorithm \ref{alg:RB_allocation}, by comparing it with the following baselines:
1) Random selection: In each time slot, the UAV draws a random set of covered devices and allocates the corresponding RBs, ensuring feasibility with \eqref{cons:P_6}-\eqref{cons:P_10}
2) Remaining data-aware selection: The UAV selects a subset of devices with the largest remaining data volumes and allocates RBs to them in each slot, subject to constraints \eqref{cons:P_6}-\eqref{cons:P_10}.
3) Channel gain-aware selection: The UAV selects a subset of devices with the highest channel gains and allocates RBs to them while satisfying constraints \eqref{cons:P_6}-\eqref{cons:P_10}.
In this evaluation, UAV control is managed by the proposed LLM-CRDT. Notably, random selection corresponds to choosing a random perfect matching of the graph $\bm{G}$ in Section \ref{subsec:resource_allocation}. The remaining data-aware and channel gain-aware schemes also construct similar bipartite graphs but compute the maximum-weight perfect matching based on their respective policies. The results are shown in Fig. \ref{fig:device_selection}, where the proposed algorithm consistently outperforms all baselines. In particular, it achieves up to a 55\% improvement in energy efficiency compared with the benchmark methods.

\begin{figure}[ht]
\centering
\includegraphics[width=0.32\textwidth]{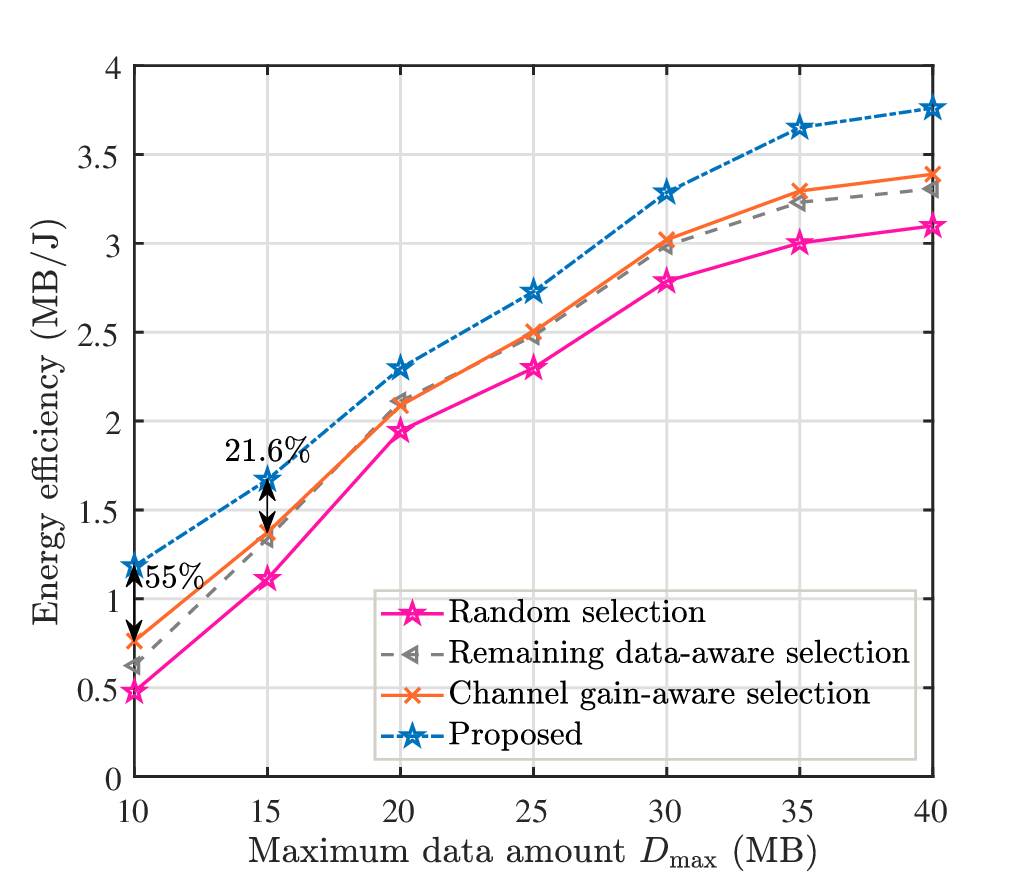}
\caption{Comparison of device selection and resource allocation approaches.}
\label{fig:device_selection}
\end{figure}

\subsection{Effects of Different Configurations for LLM-CRDT}
In this subsection, we investigate the impact of different configurations in LLM-CRDT, including the rank of the LoRA adapters, the context length, and the choice of the pre-trained LLM. Table \ref{tab:impacts_lora_rank} reports the impact of the LoRA rank on the energy efficiency performance of the proposed LLM-CRDT. Note that $r=0$ corresponds to the case where all LLM parameters are frozen (no LoRA adapters are used), while $r=d_{\rm{LLM}}$ represents full-parameter fine-tuning. Interestingly, a small LoRA rank reliably outperforms both freezing and full fine-tuning. This phenomenon can be attributed to the generalizable knowledge embedded in LLMs through large-scale pre-training, which can be effectively transferred to UAV control tasks. Since pre-trained LLMs are originally optimized on textual corpora rather than UAV trajectory optimization, simply freezing them yields limited task-specific adaptability, while full-parameter tuning risks overwriting and degrading this transferable knowledge. In contrast, selective adaptation with LoRA safeguards the pretrained model's broad generalization, enabling precise task-oriented tuning. As the LoRA adapter rank $r$ increases, both the total number of model parameters and the proportion of trainable parameters grow, inevitably incurring higher computational costs. While larger $r$ values generally improve performance, the gains diminish beyond a certain point. Specifically, when $r > 16$, the improvements become marginal. Therefore, the LoRA rank should be chosen to balance computational budget with expected gains.

\begin{table}[h]\small
\caption{Effect of LoRA Adapter Rank}
\centering
\label{tab:impacts_lora_rank}
\begin{tabular}{p{1.7cm}<{\centering}|p{2.2cm}<{\centering}|p{1.6cm}<{\centering}|p{1.5cm}<{\centering}}
\hline
LoRA $r$ & Energy efficiency(MB/J) & Model Parameters & \%Trainable Parameters\\
\hline
0   & 1.1704 &  126,148,816 &  1.35\%\\
8   & 1.9157  & 126,591,184 & 1.69\% \\
16   & 2.2975 & 127,033,552  & 2.04\% \\
24  & 2.3945 & 127,475,920 & 2.38\% \\
32  & 2.4632 & 127,918,288 & 2.72\% \\
$d_{\rm{LLM}}=768$  & 1.77825 &  338,215,812 & 100\% \\
\hline
\end{tabular}
\end{table}

In Fig. \ref{fig:LLMs_contextLen}, we explore the effect of context length and the choice of pre-trained LLM on the performance of the proposed LLM-CRDT. Specifically, we test four LLM variants from the GPT-2 series \cite{radford2019language}, whose configurations are summarized in Table \ref{tab:LLM_infos}. It is observed that energy efficiency consistently improves as the size of the pre-trained LLM increases. This is because that larger LLMs with more parameters usually possess stronger reasoning and generalization capabilities, thereby enhancing the performance of LLM-CRDT. However, the performance gap between GPT-2 XL and GPT-2 Large is smaller than that between GPT-2 Medium and GPT-2, showing that the performance gains diminish as the model size grows. Moreover, larger LLMs also incur substantially higher computational costs. Thus, the choice of LLM requires carefully balancing performance against computational efficiency. Furthermore, increasing the context length yields noticeable performance gains at first, as longer histories provide richer temporal information, allowing the model to more accurately capture trajectory correlations. Nevertheless, after the context length surpasses a certain threshold, i.e., 20, the performance gain saturates. This is because additional historical information beyond this point offers limited marginal benefit, as the system approaches its capacity to improve reward.

\begin{table}[ht]\small
\caption{Comparison of different LLMs}
\centering
\label{tab:LLM_infos}
\begin{tabular}{p{2.05cm}<{\centering}|p{1cm}<{\centering}|p{0.8cm}<{\centering}|p{1.6cm}<{\centering}|p{0.8cm}<{\centering}}
\hline
Model & \#Params & \#layer & Hidden size & \#Heads\\
\hline
GPT-2           & 124M  & 12 & 768  & 12\\
GPT-2 Medium    & 355M  & 24 & 1024 & 16\\
GPT-2 Large     & 774M  & 36 & 1280 & 20\\
GPT-2 XL        & 1.5B  & 48 & 1600 & 25\\
\hline
\end{tabular}
\end{table}

\begin{figure}[ht]
\centering
\includegraphics[width=0.32\textwidth]{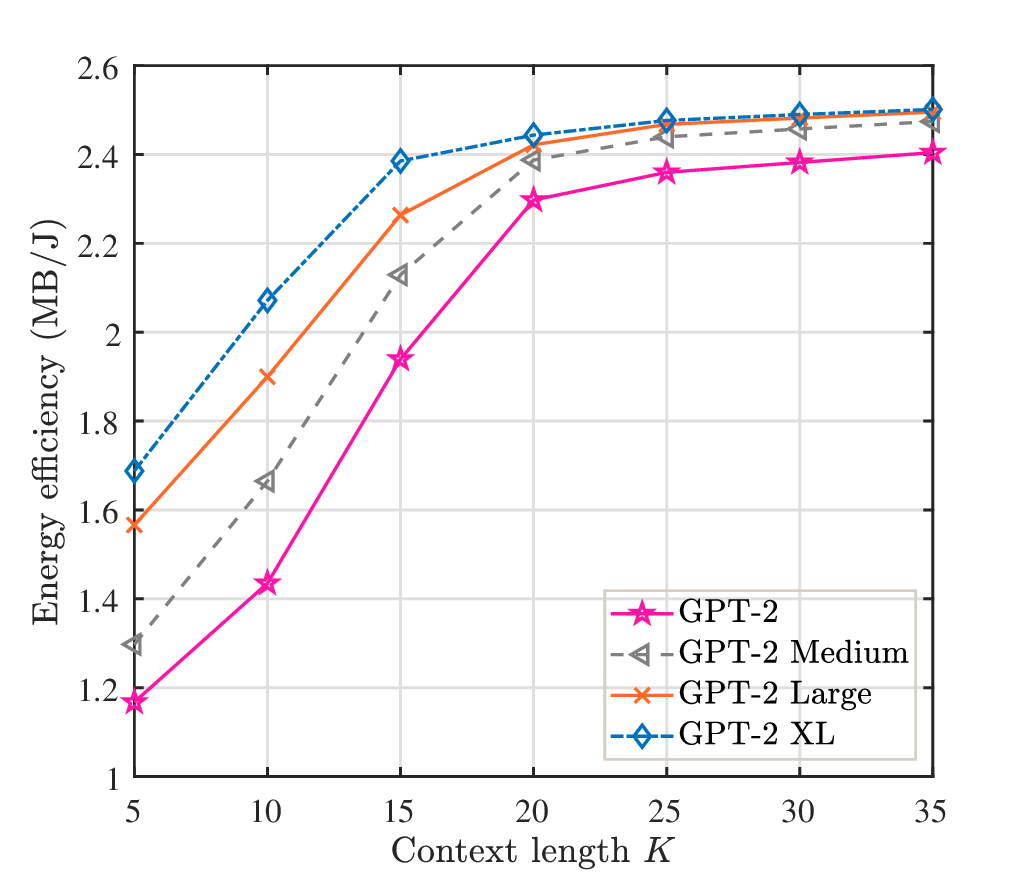}
\caption{Comparison of different device selection and resource allocation approaches.}
\label{fig:LLMs_contextLen}
\end{figure}

\section{Conclusion}\label{sec:conclusion}
This work studied the energy efficiency maximization problem for UAV-enabled data collection in IoT networks. The resource allocation subproblem was transformed into an equivalent linear programming and solved optimally, while the UAV trajectory control problem was reformulated as an offline RL problem. To overcome the reliance of expert-quality data and large-scale training in existing DT methods, we proposed an LLM-empowered critic-regularized DT framework, i.e., LLM-CRDT, to learn effective UAV control policies from a offline dataset. LLM-CRDT integrates critic-based value guidance into DT training, enabling effective policy learning from suboptimal data. Moreover, a pre-trained LLM was adopted as the transformer backbone of LLM-CRDT, and LoRA-based fine-tuning was employed to enable rapid adaptation to UAV control with small datasets and low computational overhead. Simulations demonstrated that LLM-CRDT significantly outperforms benchmark online and offline RL approaches.

\appendix
\subsection{Proof of Theorem \rm{\ref{theorem:one}}}\label{app:one}
To prove Theorem \ref{theorem:one}, we first introduce a bound that characterizes the performance gap between return-conditioned supervised learning methods (such as DT) and their underlying behavior policy.
Consider the MDP formulated in Section \ref{subsec:prob_trans}, an unknown behaviour policy $\pi_{\rm{B}}$, and a conditioning function $f$.
Assuming the following conditions:
\begin{enumerate}
  \item Let $g(\bm{\tau}) = \sum\nolimits_{t = 1}^T r_t$ denotes the total return of a trajectory $\bm{\tau}$. For all initial state $\bm{s}_1 \in \mathcal{S}$, the conditioning function satisfy $\Pr\left(g({\bm{\tau}}) = f(\bm{s}_1)\left| \bm{s}_1 \right.\right) \ge \zeta$.
  \item The MDP is $\varepsilon$-deterministic, i.e., $\Pr\left(r_t \ne \mathcal{R}(\bm{s}_t, \bm{a}_t)~{\rm{or}}~ \bm{s}_{t+1} \ne \mathcal{P}(\bm{s}_t, \bm{a}_t)\left| \bm{s}_t, \bm{a}_t \right. \right) \le \varepsilon$  at all $\bm{s}_t$ and $\bm{a}_t$ for the reward function $\mathcal{R}$ and environment dynamics $\mathcal{P}$.
  \item The conditioning function is consistent of the RTG, i.e., $f(\bm{s}_t) = f(\bm{s}_{t+1}) + r_t$ for all $\bm{s}_t$.
\end{enumerate}
Then, the RCSL methods trained on data generated by the behavior policy $\pi_{\rm{B}}$ satisfies the following performance bound:
\begin{align}\label{eq:RCSL_bound}
\mathbb{E}_{\bm{s}_1}\left[ f(\bm{s}_1) \right] - \mathbb{E}_{\bm{\tau}  \sim \pi_f^{\rm{RCSL}}}\left[ g(\bm{\tau}) \right] \le \varepsilon \Big( \frac{1}{\zeta} + 2 \Big) T^2,
\end{align}
where $\pi_f^{\rm{RCSL}}$ represents the policy learned by RCSL methods when conditioned on function $f$.
The proof of this bound is provided in \cite{NEURIPS2022_0a2f65c9}.
Given an offline datset $\mathcal{D}$ collected by the behaviour policy $\pi_{\rm{B}}$, we choose the conditioning function as $f(\bm{s}_1) = \sum\nolimits_{r_{1:T} \sim \pi_{\rm{B}}(\bm{s}_1)} r$ and substitute it into \eqref{eq:RCSL_bound}, we have
\begin{align}
&\mathbb{E}_{\bm{s}_1}\left[ f(\bm{s}_1) \right] - \mathbb{E}_{\bm{\tau}  \sim \pi_f^{\rm{RCSL}}}\left[ g(\bm{\tau}) \right] \notag\\
&= \mathbb{E}_{\bm{s}_1}\Big[ \sum\nolimits_{r_{1:T} \sim \pi_{\rm{B}} (\bm{s}_1)} r  \Big] - \mathbb{E}_{\bm{\tau}  \sim \pi_f^{\rm{RCSL}}}\left[ g(\bm{\tau}) \right] \notag\\
&= \mathbb{E}_{\bm{\tau} \sim \pi_{\rm{B}}}\Big[ \sum\nolimits_{t = 1}^T r_t \Big] - \mathbb{E}_{\bm{\tau}  \sim \pi_f^{\rm{RCSL}}}\left[ g(\bm{\tau}) \right] \notag\\
&= \mathbb{E}_{\bm{\tau} \sim \pi_{\rm{B}}}\left[ g(\bm{\tau}) \right] - \mathbb{E}_{\bm{\tau}  \sim \pi_f^{\rm{RCSL}}}\left[ g(\bm{\tau}) \right].
\end{align}
Then recall the reward-to-go $\hat{R}_t = \sum\nolimits_{i = t}^T r_i$, it is obvious that the conditioning function $f(\bm{s}_t)=\hat{R}_t$ satisfy the requirement about the consistence of conditioning function, i.e., condition 3). Thus, we have
\begin{align}
\mathbb{E}_{\bm{s}_1}\left[ f(\bm{s}_1) \right] \!-\! \mathbb{E}_{\bm{\tau}  \sim \pi_f^{\rm{RCSL}}}\left[ g(\bm{\tau}) \right]
&\!=\! \mathbb{E}_{\bm{\tau} \sim \pi_{\rm{B}}}\left[ g(\bm{\tau}) \right] \!-\! \mathbb{E}_{\bm{\tau} \sim \pi_{\bm{\theta}}}\left[ g(\bm{\tau}) \right] \notag\\[-0.1cm]
&\le \varepsilon \Big( \frac{1}{\zeta} + 2 \Big) T^2,
\end{align}
where $\pi_{\bm{\theta}}$ is the learned policy of DT.
Thus, the proof is completed.

\scriptsize
\bibliographystyle{IEEEtran}
\bibliography{IEEEabrv,cited}
\end{document}